\documentclass[a4paper,UKenglish,cleveref, autoref, numberwithinsect, thm-restate]{lipics-v2021}

\usepackage{amsmath}
\usepackage{amsthm}
\usepackage{proof}
\usepackage{graphicx}
\usepackage{xspace}
\usepackage{color} 
\usepackage[all]{xy}
\usepackage{tikz-cd}
\usepackage{hyperref}
\usepackage{mathtools}
\usepackage{macros}
\usepackage{MnSymbol}
\usepackage[shortlabels]{enumitem}
\usepackage{booktabs}
\usepackage{relsize} 
\usepackage{float}

\usepackage{todonotes}

\title{Function spaces for orbit-finite sets} 

\author{Miko{\l}aj Boja\'nczyk}{University of Warsaw}{bojan@mimuw.edu.pl}{https://orcid.org/0000-0002-1825-0097}{}

\author{Lê Thành Dũng (Tito) Nguy\~{\^{e}}n}{École normale supérieure de Lyon, France \and \url{https://nguyentito.eu/}}{nltd@nguyentito.eu}{https://orcid.org/0000-0002-6900-5577}{Supported by the LABEX MILYON (ANR-10-LABX-0070) of Université de Lyon, within the program ``France 2030'' (ANR-11-IDEX-0007) operated by the French National Research Agency (ANR).}

\author{Rafa{\l} Stefa\'nski}{University of Warsaw}{rafal.stefanski@mimuw.edu.pl}{https://orcid.org/0000-0002-1825-0097}{}

\authorrunning{M. Boja\'nczyk, L. T. D. Nguy{\^{e}}n and R. Stefa\'nski } 

\Copyright{Miko{\l}aj Boja\'nczyk, L{\^{e}} Th{\`{a}}nh Dung Nguy{\^{e}}n and Rafa{\l} Stefa\'nski } 



\category{} 

\relatedversion{} 

\acknowledgements{L.~T.~D.\ Nguy\~{\^{e}}n would like to thank Clovis Eberhart and Cécilia Pradic for their ongoing collaboration on ``implicit automata for data words'' (cf.~\cite[\S1.4.4]{titoPhD}), which inspired some ideas in this paper.}

\nolinenumbers 
\hideLIPIcs

\EventEditors{John Q. Open and Joan R. Access}
\EventNoEds{2}
\EventLongTitle{42nd Conference on Very Important Topics (CVIT 2016)}
\EventShortTitle{CVIT 2016}
\EventAcronym{CVIT}
\EventYear{2016}
\EventDate{December 24--27, 2016}
\EventLocation{Little Whinging, United Kingdom}
\EventLogo{}
\SeriesVolume{42}
\ArticleNo{23}

\begin{document}

\newcommand{\loli}{\multimap}

\maketitle 
\begin{abstract}
    Orbit-finite sets are a generalisation of finite sets, and as such support many operations allowed for finite sets, such as pairing, quotienting, or taking subsets. However, they do not support function spaces, i.e.~if $X$ and $Y$ are orbit-finite sets, then the space of finitely supported functions from $X$ to $Y$ is not orbit-finite. In this paper we propose two solutions to this problem: one is obtained by generalising the notion of orbit-finite set, and the other one is obtained by restricting it. In both cases, function spaces and the original closure properties are retained. Curiously, both solutions are ``linear'': the generalisation is based on linear algebra, while the restriction is based on linear logic.
\end{abstract}

\section{Introduction}
The class of orbit-finite sets is a class of sets that contains all finite sets and some infinite sets, but still shares some properties with the class of finite sets.  The idea, which dates back to Fraenkel--Mostowski models of set theory,  is to begin with an infinite set $\atoms$ of \emph{atoms} or \emph{urelements}. We think of the atoms as being names, such as Eve or John, and atoms can only be compared with respect to equality. Intuitively speaking, an  orbit-finite set is a set  that can be constructed using the atoms, such as $\atoms^2$ or $\atoms^*$, subject to the constraint that there are finitely many elements up to  renaming atoms. For example, $\atoms^2$ is orbit-finite because it has two elements up to renaming atoms, namely (John, John) and (John, Eve), while $\atoms^*$ is not orbit-finite, because the length of a sequence is invariant under renaming atoms, and there are infinitely many possible lengths. For a survey on orbit-finite sets, see~\cite{bojanczyk_slightly2018}.

The notion of orbit-finiteness can be seen as an attempt to find an appropriate notion of finiteness for the  nominal sets of Gabbay and Pitts~\cite{PittsAM:nomsns}.  This attempt emerged from the study of computational models such as monoids~\cite{bojanczykNominalMonoids2013} and automata~\cite{bojanczykAutomataTheoryNominal2014} over infinite alphabets. Since automata are also the main use case for the present paper, we illustrate orbit-finiteness using an automaton example. 

\begin{example}[An orbit-finite automaton]\label{ex:first-letter-repeats}
    Let  $L \subseteq \atoms^*$ be the language of  words in which the letter from the first position does not appear again. This language contains John $\cdot$ Mark $\cdot$ Mark $\cdot$ Eve, because John does not reappear, but it does not contain John $\cdot$ Mark $\cdot$ John. To recognize this language, we can use a deterministic automaton, which uses its state to remember the first letter. In this automaton, the input alphabet is $\Sigma = \atoms$ and  the state space is $Q = 1 + 1 + \atoms$. In this state space, there are two special states, namely the initial state and a rejecting error state, and furthermore there is one state for each atom $a \in \atoms$, which represents a situation where the first letter was $a$ but it has not been seen again yet. This state space is infinite; but it is orbit-finite, since each of the three components in $Q$ represents a single orbit, with the orbits of the $1$ components being singletons.\exampleend
\end{example}

Orbit-finite sets have many advantages, which ensure that they are a good setting for automata theory, and discrete mathematics in general. For example, an orbit-finite set can be represented in a finite way~\cite{bojanczyk_slightly2018}, which ensures that it becomes meaningful to talk about algorithms that input orbit-finite sets, such as an emptiness check for an automaton. Also, orbit-finite sets are closed under taking disjoint unions and products, which ensures that natural automata constructions, such as the union of two nondeterministic automata or the product of two deterministic automata can be performed.

However, orbit-finite sets do not have all the closure properties of finite sets. Notably missing is the powerset construction, and more generally taking function spaces. For example, if we look at the powerset of $\atoms$, then this powerset will not be orbit-finite, since already the finite subsets give infinitely many orbits (two finite subsets of different sizes will be in different orbits). The lack of powersets means that one cannot do the subset construction from automata theory, and in particular deterministic and nondeterministic automata are not equivalent. This non-equivalence was known from the early days of automata for infinite alphabets~\cite{kaminskiFiniteMemoryAutomata1994}, and in fact, some decision problems, such as equivalence, are decidable for deterministic automata but undecidable for nondeterministic automata~\cite{nevenFiniteStateMachines2004}. Another construction that fails is converting a deterministic automaton into a monoid~\cite[p.~221]{bojanczykNominalMonoids2013}; this is because function spaces on orbit-finite sets are no longer orbit-finite, as explained in the following example. 

\begin{example}[Failure of the monoid construction]\label{ex:first-letter-repeats-monoid}
    Let us show that the  automaton from Example~\ref{ex:first-letter-repeats} cannot be converted into a monoid. The standard construction would be to define the monoid as the subset $M \subseteq Q \to Q$ of all state transformations, namely the subset generated by individual input letters. 
    Unfortunately, this construction does not work. This is because in order  for  two input words to give the same state transformation, they need to have the same set of letters that appear in them. In particular, the corresponding set of set transformations is not orbit-finite, for the same reason as why the finite powerset is not orbit-finite. Not only does the standard construction not work, but also this language is not recognized by any orbit-finite monoid.\exampleend
\end{example}

An attempt to address this problem was provided in~\cite{stefansk-msc,stefanski-phd,bojanczykstefanski2020}, by using \emph{single-use} functions. The idea, which originates in linear types and linear logic,  is to restrict the functions so that they use each argument at most once. For example, consider the following two functions that input atoms and output Booleans:
\begin{align*}
a \in \atoms \mapsto 
\begin{cases}
    \text{true} & \text{if $a=$ John}\\
    \text{false} & \text{otherwise}
\end{cases}
\qquad \qquad 
a \in \atoms \mapsto 
\begin{cases}
    \text{true} & \text{if $a=$ John or $a = $Eve}\\
    \text{false} & \text{otherwise}
\end{cases}
\end{align*}
The first function is single-use, since it compares the input atom to John only, while the second function is not single-use, since it requires two comparisons, with John and Eve. Here is another example, which shows that the problems with the monoid construction from Example~\ref{ex:first-letter-repeats-monoid} could be blamed on a violation of the single-use condition.
\begin{example}
    Consider  the transition function of the automaton in Example~\ref{ex:first-letter-repeats}, which inputs a state in $1 + 1 + \atoms$ together with an input letter from $\atoms$, and returns a new state.  This function is not single use. Indeed, if the state is in $\atoms$, then the transition function compares it for equality with the input letter; but if the comparison returns true, another copy of the old state must be kept as the new state for future comparisons. \exampleend
\end{example}

If one restricts attention to functions that are single-use, much of the usual robustness of automata theory is recovered, with deterministic automata being equivalent to monoids, and both being equivalent to two-way deterministic automata~\cite{bojanczykstefanski2020}.

Despite the success of the single-use restriction in solving automata problems, one would ideally prefer a more principled approach, in which instead of defining single-use automata, we would define a more general object, namely single-use sets and functions. Then the definitions of  automata and monoids, as well theorems speaking about them, should arise automatically as a result of suitable closure properties of the sets and functions.

This approach was pursued in~\cite{stefanski-phd}, in which a \emph{category} of orbit-finite sets with single-use functions was proposed. In this corresponding category,  one can represent the set of all single-use functions between two orbit-finite sets $X$ and $Y$ as a new set, call it $X \Rightarrow Y$, which is  also orbit-finite.  However, as we will see later in this paper, this proposal is not entirely satisfactory, since it fails to account for standard operations that one would like to perform on function spaces, most importantly partial application (currying). Note that partial application is crucial for converting an automaton into a monoid, since the monoid consists of partially applied transition functions, in which the input word is known, but the input state is not. In the language of category theory, the proposal from~\cite{stefanski-phd} failed to be a monoidal closed category.

Let us mention there have been several works using category theory to generalize classical operations on automata, such as the coalgebraic ``generalized powerset construction''~\cite{DBLP:journals/corr/abs-1302-1046}. The closest to the philosophy we exposed above might be the setting introduced by Colcombet and Petrişan~\cite{colcombet2020automata}, where automata in different categories are compared (see~\cite{ColcombetPS21,Aristote24} for applications). Within this setting, Nguy{\~{ê}}n and Pradic have studied some properties of automata over monoidal closed categories~\cite[Sections~1.2.3~and~4.7--4.8]{titoPhD} as part of their research on ``implicit automata''\footnote{\label{ftn:iatlc}This line of work is about relating the expressive power of automata and typed $\lambda$-calculi. In~\cite{IATLC2,titoPhD}, in addition to results on arbitrary symmetric monoidal closed categories, such a category of single-use assignments on string-valued registers is built and used to relate a register-based string transducer model to a $\lambda$-calculus with linear types. Indeed, symmetric monoidal closed categories are famous for providing denotational semantics for the linear $\lambda$-calculus. Similarly, our results here could serve to characterize the languages of words with atoms from~\cite{bojanczykstefanski2020} via some typed $\lambda$-calculus.}~\cite{IATLC,IATLC2,titoPhD,pradic2024implicit}.

 \subparagraph{Contributions of this paper.}
 The main contribution of this paper is to propose a notion of single-use sets and functions, which extends the proposal from~\cite{stefanski-phd}, but which is rich enough to be closed under taking function spaces. More formally, we propose a category of single-use functions between orbit-finite sets equipped with additional metadata, and we prove that this category is symmetric monoidal closed.

The main idea is to follow the tradition of linear types, and to  distinguish two kinds of products, namely\footnote{Categorically, $\&$ is the cartesian product while $\otimes$ is a monoidal product.} $X \otimes Y$ and $X \& Y$. Thanks to this distinction, the function space can be built so that the appropriate operations on functions, namely application and currying, can be implemented in a single-use way.

Our proposed category is strongly inspired by linear types, and the proof that it is symmetric monoidal closed uses a form of \emph{game semantics}\footnote{In this paper, the game semantics will only appear in the appendix, as they are part of our proofs rather than our main claims. That said, let us point to the lecture notes~\cite{abramsky2013semantics,Hyland1997} as references for the category of ``simple games'' upon which we build. For a recent survey of modern game semantics, see~\cite{ClairambaultHDR}.\\
  In the work mentioned in \Cref{ftn:iatlc}, the fact that ``register contents do not influence the control flow'' leads to a much simpler construction, similar to Dialectica categories~\cite{PaivaDialectica} (which can be seen as a form of game semantics with very short plays), where no quotient is needed. Our \Cref{thm:single-use-closed} might also be provable by representing single-use functions as $\lambda$-terms (or programs in some richly structured syntactic formalism) instead of strategies over games. Indeed, it is a classical fact that a simple type is inhabited by finitely many linear $\lambda$-terms up to $\beta$-conversion, and variations on this fact have been used in the literature to relate automata and $\lambda$-calculus~\cite{IATLC,LambdaTransducer}.} -- a tool that we take from programming language theory. However, to the best of our knowledge, it is an original idea to have an infinite but orbit-finite base type $\atoms$, and to observe that all constructions in game semantics are consistent with orbit-finiteness. We believe that the resulting category deserves further study, and that it is an interesting and non-trivial example of a category representing ``finite'' objects.

Along the way, we provide examples of how the category can be useful in automata theory. Our main example is converting a deterministic automaton into a monoid, but another example is modelling two-way deterministic automata, and deciding their emptiness. 

An important property of our construction is that it is generic. In fact, instead of having a set that is equipped with equality only, one could apply the construction to any relational structure, e.g.~real arithmetic $(\mathbb R, +, \times, <)$, as long as the structure is given using relations and not functions. Under further assumptions on the structure, such as having a decidable first-order theory (which holds for real arithmetic) and being oligomorphic (which does not hold for real arithmetic but does hold for the rational numbers with their linear order), further benefits in the resulting single-use category can be derived, such as having an emptiness algorithm for deterministic two-way automata. 

Finally, as a minor contribution, we present an alternative solution for the problem of function spaces, which is  to use vector spaces of orbit-finite dimension. This is a minor contribution as far as the present paper is concerned, because the technical tools were developed already in~\cite{bojanczykKM21OrbitFiniteVector}, and the only contribution -- if any -- of this paper is one of perspective, namely framing it as a symmetric monoidal closed category. An advantage of the vector space category is its simplicity, and the fact that it is ``bigger'' in the following sense. The two solutions for function spaces discussed in this paper, namely the single-use solution and the vector space solution, sit on both sides of the classical category of orbit-finite sets, as witnessed by two faithful functors, one from the single-use category to the orbit-finite category, and one from the orbit-finite category to the vector space category. 
The generality of vector spaces comes at a price, though. As mentioned before, the single-use construction can be applied to any structure, and the orbit-finite benefits can be derived for all $\omega$-categorical structures (this is the standard assumption in the study of orbit-finiteness). In contrast, the vector space construction is more brittle, and it fails for certain $\omega$-categorical  structures such as the Rado graph.   Another disadvantage of the vector space category is that it is not traced with respect to the coproduct, unlike the single-use category, which precludes the applications to two-way automata.

\section{Sets with atoms}
We begin  with  a brief introduction to orbit-finite sets. For a more detailed treatment, see~\cite{bojanczyk_slightly2018}.

Fix for the rest of this paper a countably infinite set $\atoms$, whose elements will be called atoms.  We assume that this set has no other structure except for equality, which will mean that we will only be interested in notions which are \emph{equivariant}, i.e.~invariant under renaming atoms. For example, $\atoms$ has only two equivariant subsets, namely the empty and full subsets. On the other hand, the set $\atoms^2$ has four equivariant subsets; this is because any subset containing (Eve, Eve) must contain all pairs on the diagonal, and any subset containing (Eve, John) must contain all pairs outside the diagonal.  In order to meaningfully speak about equivariant subsets, we must be able to have an action of atom renamings on the set, as formalized in the following definition. The finite support condition is a technical condition that ensures that the action is well-behaved; this condition dates back to the work of Fraenkel and Mostowski, and is explained in the survey texts~\cite{PittsAM:nomsns,bojanczyk_slightly2018}.



\begin{definition}[Set with atoms]
    A \emph{set with atoms} is a set $X$, equipped with an action of the group of atom renamings, subject to the following \emph{finite support condition}: for every $x \in X$ there is a finite set of atoms, such that if an atom renaming $\pi$ fixes all atoms in the set, then it also fixes $x$.
\end{definition}

The idea is that a set with atoms is any kind of object that deals with atoms, such as the set $\atoms^*$ of all words over the alphabet $\atoms$, or the family of finite subsets of $\atoms$. Among such objects, we will be interested in those which are ``finite''. This will be formalized by  saying that there are  finitely many orbits, as described below.
Define the \emph{orbit} of an element $x$ in a set with atoms to be the elements that can be obtained from $x$ by applying some atom renaming. For example, in the set $\atoms^2$,  the orbit of (John, Eve)   contains  (Mark, John), but it does not contain (John, John). The orbits form a partition of a set with atoms.

\begin{definition}[Orbit-finite set]   A set with atoms is called \emph{orbit-finite} if it has finitely many orbits. 
\end{definition}

A typical example of an orbit-finite set is $\atoms^4$, or more generally any polynomial expression such as $\atoms^4 + \atoms^3 + \atoms^3 + 1$. Here, $1$ represents the set of zero-length sequences; this set has a unique element which is its own orbit.  For  example,  $\atoms^3$ has five orbits, because there are five possible ways of choosing a pattern of equalities in a sequence of three names. On the other hand,  $\atoms^*$ has infinitely many orbits, since sequences of different lengths are necessarily in different orbits.   The family of finite subsets of $\atoms$  is also  not orbit-finite, because subsets of different sizes are in different orbits. The full powerset $\powerset \atoms$ is not even a legitimate object in our setting, because some of its elements, i.e.~some subsets of $\atoms$, violate the finite support condition.

        

\subsection{Finiteness of function spaces}
\label{sec:orbit-finite-function-spaces}
As mentioned above, orbit-finite sets can be seen as  a certain generalization of finite sets. They allow some, but not all, operations that can usually be done on finite sets. For example, orbit-finite sets are closed under disjoint unions $X + Y$ and products $X \times Y$.  Another good property is that an orbit-finite set has only finitely many equivariant subsets (an equivariant subset is one that is invariant under the action of atom permutations). This is because an equivariant subset is a union of some of the finitely many orbits. This accounts for some of the good computational properties of orbit-finite sets. For example, nonemptiness is decidable for orbit-finite automata (see below), because the state space can be searched orbit by orbit.
However, orbit-finite sets do not have orbit-finite function spaces, as explained in the following example. 

\begin{example}
    Consider the space of functions of type $\atoms \to \atoms$. What are the legitimate functions? One choice is that we only allow the equivariant functions. Under this choice, there is only one possible function,  namely the identity function. However, as we will explain below, the more appropriate choice is the class of finitely supported functions, i.e.~those that are invariant under all atom renamings that fix some finite set of atoms that depends only on the function. For example,  the  function
    \begin{align*}
    f(a) = \begin{cases}
        \text{Mark} & \text{if $a \in \set{\text{John, Eve, Bill}}$} \\
        a & \text{otherwise}
    \end{cases}
    \end{align*}
    is finitely supported, because it is invariant under all atom renamings that fix Mark, John, Eve and Bill. The space of finitely supported functions will be called the \emph{finitely supported function space}. This space is equipped with an action of atom renamings. For example, if $\pi$ is the atom renaming that swaps Mark with Adam, then applying it to the function $f$ defined above gives the function $\pi(f)$ that has the same definition (or source code, if a programming intuition is to be followed), except that Mark is used instead of Adam.

    In the case of $\atoms \to \atoms$, the finitely supported function space is not orbit-finite. Indeed, the condition $a \in \set{\text{John, Eve, Bill}}$ can be replaced by $a \in X$ for any finite set $X \subset \atoms$ of exceptional values, and two choices of $X$ of different cardinalities will give us two functions in different orbits. \exampleend
\end{example}

In the above example, we explained how the finitely supported function space is not always orbit-finite. On the other hand, the equivariant function space will always be literally finite, if the input and output types are orbit-finite, because it will be an equivariant subset of the product $X \times Y$. So why do we insist on the finitely supported function space? A more principled reason will appear later in the paper, where we discuss symmetric monoidal closed categories, and where the finitely supported function space will turn out to be the right one. However, we can already give a simple reason, which appeals to automata theory, namely converting an automaton into a monoid. As we have seen in Example~\ref{ex:first-letter-repeats-monoid}, when converting an automaton to a monoid, we will want to use partially applied transition functions, and such functions will be finitely supported but not equivariant.

The lack of function spaces is the problem  addressed in this paper. We present two solutions of this problem. In each of the two solutions, we modify the notion of orbit-finite sets, by either restricting it or generalizing it, in a way that recovers function spaces.

\section{Single-use sets and functions}
\label{sec:single-use-sets}
We now turn to our first solution for the problem of orbit-finite function spaces.  Our solution builds on the idea from~\cite[Section 2.2]{stefanski-phd}, which is to consider only functions that are single-use. We describe this idea in Section~\ref{sec:single-use-functions-over-polynomial-orbit-finite-sets}, and we show how it almost, but not quite, achieves function spaces. Then, in the rest of this section, we show how function spaces can be recovered by using a more refined type system. 

\subsection{Single-use functions over polynomial orbit-finite sets}
\label{sec:single-use-functions-over-polynomial-orbit-finite-sets}

In the introduction, we have already given an intuitive description behind the single-use functions; these are functions that destroy the argument after any use of it, such as comparing it to a constant. We now give a more formal definition.

We do not define the single-use functions on all orbit-finite sets, but only a syntactically defined fragment, namely the \emph{polynomial orbit-finite sets}, which are sets that can be generated from $1$ and $\atoms$ using  
 products $\times$ and disjoint unions $+$. Therefore, we will allow orbit-finite sets like $1 + \atoms^2$, but we will not allow orbit-finite sets like 
the set of non-repeating pairs
$\setbuild{(a,b)}{$a \neq b \in \atoms$}$ or the set of unordered pairs $\setbuild{\set{a,b}}{$a \neq b \in \atoms$}$. It is an open problem to find a satisfactory definition of single-use functions on all orbit-finite sets. (A simple hack is to use a quotienting construction, similar to \Cref{sec:quotient-category}, but what we would really like to do is to identify some extra structure in a set, possibly an action of some yet unknown group or semigroup, which enables us to speak about single-use functions.)

Consider two polynomial orbit-finite sets $X$ and $Y$. To define which functions $X \to Y$ are single-use, we use an inductive definition. We begin with certain functions that are considered single-use, such as the equality test of type $\atoms \times \atoms \to 1 + 1$. These functions are called the \emph{prime functions}, and their full list is given in Figure~\ref{fig:prime-morphisms-without-with}. Next, we combine the prime functions into new ones using three combinators. The first, and most important, combinator is  function composition. Then, we have two combinators for the two type constructors: if two functions $f_1 : X_1 \to Y_1$ and $f_2 : X_2 \to Y_2$ are single-use, then the same is true for:
\begin{align*}
    f_1 \times f_2 : X_1 \times X_2 \to Y_1 \times Y_2
    \qquad &
    f_1 + f_2 : X_1 + X_2 \to Y_1 + Y_2 \\
    \scriptstyle  (x_1,x_2) \mapsto (f_1(x_1),f_2(x_2)) 
    \qquad &
    {\scriptstyle \text{left}(x_1) \mapsto \text{left}(f_1(x_1)) 
        \quad 
        \text{right}(x_2) \mapsto \text{right}(f_2(x_2)) }.
\end{align*}

Crucially, the list of prime single-use functions does not contain the copying function $a \in \atoms \mapsto (a,a) \in \atoms^2$. Therefore, an alternative name for the single-use functions is \emph{copyless}. If we added copying, then we would get all finitely supported functions~\cite[Lemma 23]{stefanski-phd}.

\begin{table}[h!]
    \centering
    \begin{tabular}{lll}
        \textbf{Function} & \textbf{Type} & \textbf{Definition} \\ \\
        \emph{Functions about $\atoms$} \\
        equality test & $\atoms \times \atoms \to 1 + 1$ & $a, b \mapsto \text{if } a = b \text{ then true else false}$ \\
        constant $a$ & $1 \to \atoms$ & $x \mapsto a$ \\
        identity & $\atoms \to \atoms$ & $x \mapsto x$ \\
        \\
        \emph{Functions about \(\times\)} \\
        commutativity of $\times$ & $X \times Y \to Y \times X$ & $x \times y \mapsto y \times x$ \\
        first projection & $X \times Y \to X$ & $x \times y \mapsto x$ \\
        second projection & $X \times Y \to Y$ & $x \times y \mapsto y$ \\
        append 1 & $X \to X \times 1$ & $x \mapsto x \times ()$ \\
        associativity of $\times$ & $(X \times Y) \times Z \to X \times (Y \times Z)$ & $(x \times y) \times z \mapsto x \times (y \times z)$ \\ \\
        \emph{Functions about \(+\)} \\
        first co-projection & $X \to X + Y$ & $x \mapsto \text{left}(x)$ \\
        second co-projection & $Y \to X + Y$ & $y \mapsto \text{right}(y)$ \\
        co-diagonal & $X + X \to X$ & $\left\{\begin{tabular}{l}
            $\text{left}(x) \mapsto x$\\
            $\text{right}(x) \mapsto x$
            \end{tabular}\right.$ \\
        commutativity of $+$ & $X + Y \to Y + X$ & $\left\{\begin{tabular}{l}
        $\textrm{left}(x) \mapsto \textrm{right}(x)$\\
        $\textrm{right}(y) \mapsto \textrm{left}(y)$
        \end{tabular}\right.$ \\
        associativity of $+$ & $(X + Y) + Z \to X + (Y + Z)$ & $\left\{
        \begin{tabular}{l}
        $\text{left}(\text{left}(x)) \mapsto \text{left}(x)$\\
        $\text{left}(\text{right}(y)) \mapsto \text{right}(\text{left}(y))$\\
        $\text{right}(z)\mapsto \text{right}(\text{right}(z))$
        \end{tabular}\right.$ \\
        \\
        \emph{Distributivity}
        \\
        $+$ distributes over $\times$ & $X \times (Y + Z) \to (X \times Y) + (X \times Z)$ & $\left\{\begin{tabular}{l}
            $x \times (\text{left}(y)) \mapsto \text{left}(x \times y)$\\
            $x \times (\text{right}(z)) \mapsto \text{right}(x \times z)$
        \end{tabular}\right.$ \\
        \\
    \end{tabular}
    \caption{The prime single-use functions for polynomial orbit-finite sets $X, Y$ and $Z$.}
    \label{fig:prime-morphisms-without-with}
\end{table}

\begin{example}\label{ex:six-compositions}
    Consider function of type $\atoms^3 \to \atoms$ which inputs a triple $(a,b,c)$ of atoms and returns $a$ if $c$ is equal to Mark, and $b$ otherwise. This function is a single-use function. It is obtained by composing the six functions listed below:
\begin{center}
    \begin{tabular}{ll}
        Function & Type after function \\
        \hline
        Append 1. & $ \atoms \times \atoms \times \atoms \times 1$ \\
        Replace added $1$ with Mark using the constant function. & $\atoms \times \atoms \times \atoms \times \atoms$ \\
        Apply the equality test to the last two components. & $ \atoms \times \atoms \times (1+1)$ \\
        Distribute. & $ \atoms \times \atoms \times 1 +   \atoms \times \atoms \times 1$ \\
        Project to first and second components, respectively. & $\atoms + \atoms$ \\
        Co-diagonal & $\atoms$ 
    \end{tabular}
\end{center}
To justify this description, one should also show that the six functions are single-use. Three of the functions, namely append 1, distributivity and co-diagonal are prime functions. The other three are obtained by combining prime functions using the combinators. For example, the equality test is paired, using the combinator for $\times$, with the identity on the remaining two atoms. \exampleend
\end{example}

The design goal of the single-use restriction is to have orbit-finite function spaces. The rough idea is that a single-use function can only use a bounded number of atoms in its source code, which guarantees orbit-finiteness of the function space. 

\begin{example}\label{ex:first-single-use-function-space}
    Consider function of type $\atoms \to 1 + 1$, which can be seen as subsets of the atoms, with $1+1$ representing the Booleans. We will consider two function spaces: the larger space of all finitely supported functions, and the smaller space of single-use functions.
    
    A function in the larger space is any finitely supported subset of the atoms; such subsets are the same as the finite and co-finite subsets. Therefore, the larger function space admits an equivariant bijection with a disjoint union of two copies of the finite powerset $\powerset_{\text{fin}} \atoms$, in particular it is not orbit-finite.
    
    Consider now the smaller single-use function space. There are four possible functions of this kind: (a) always return true; (b) always return false; (c) check for equality with some fixed atom $a$; (d) check for disequality ($\neq$) with some fixed atom $a$.  Therefore,  the set of single-use functions  admits an equivariant bijection with the orbit-finite set $1 + 1 + \atoms + \atoms$.
    \exampleend
\end{example}

The above example shows that the space of single-use functions of some type $X \to Y$ is orbit-finite, and in fact it can be described using a polynomial orbit-finite set. This is true for every choice of polynomial orbit-finite sets $X$ and $Y$, as proved in~\cite[Theorem 5]{stefanski-phd}, and illustrated in the following example. 

\begin{example}\label{ex:decision-tree-types} Assume that the input type $X$ is some power of the atoms $\atoms^k$, and the output type $Y$ does not use atoms, e.g.~it is~$Y = 1 +1$. The assumption on the input type can be made without loss of generality using distributivity, while the assumption on the output type is a proper restriction, but it will allow us to skip some technical details of the general construction while retaining the important intuitions. 
We describe below a type that represents all single-use functions from $\atoms^k$ to $Y$; we shall denote it by $\atoms^k \Rightarrow Y$. Note that $\Rightarrow$ is \emph{not} a primitive type constructor in our grammar of types; it is a notation that stands for the inductive construction below.

This type is defined by  induction on $k$. In the base case of $k=0$ we simply need to give a value from the output type, and therefore $\atoms^0 \Rightarrow Y$ is the same as $Y$. Consider now the induction step of $k > 0$.  
    We observe that a single-use function that inputs $\atoms^k$ must begin with some equality test, and then continue with one of two single-use functions that have fewer arguments (one for the case when the equality test returns true, and one for the other case). This observation leads to the following definition of the type $\atoms^k \Rightarrow Y$:
\begin{align*}
\myunderbrace{ \coprod_{i \in \set{1,\ldots,k}} 
    \atoms \times (\atoms^{k-1} \Rightarrow Y) \times (\atoms^{k-1} \Rightarrow Y)
 }{starts by comparing $i$-th  \\
 coordinate to some constant}  \quad + \quad 
\myunderbrace{\coprod_{i, j \in \set{1,\ldots,k}} (\atoms^{k-2} \to Y) \times (\atoms^{k-2} \to Y)}{
    starts by comparing the \\ 
    $i$-th and $j$-th coordinates}.
\end{align*}
Note that the above representation of the function space is not necessarily unique, i.e.~the same function can be represented in several different ways. For example, the order in which equality tests are performed will matter for the representation, but might not matter for the function. This is not something that we worry about, and we will use function spaces with non-unique representations in the paper, see also \Cref{sec:quotient-category} for how we deal with non-uniqueness. \exampleend
\end{example}

\subparagraph{Problem with currying.}
Unfortunately, the proposal illustrated in Example~\ref{ex:decision-tree-types} and described in more detail in~\cite{stefanski-phd}  does not give a satisfactory solution to the problem of function spaces. The problem is that the set of representations $X \Rightarrow Y$ should  also support operations on functions. More specifically, we should be able to indicate single-use operations which do the following:
\begin{description}
    \item[evaluation:] a single-use function from $(X \Rightarrow Y) \times X$ to $Y$ which inputs a representation of a function and applies it to an argument;
    \item[composition:] a function from $(X \Rightarrow Y) \times (Y \Rightarrow Z)$ to $(X \Rightarrow Z)$ which inputs the representations of two functions and returns a representation of their composition.
    \item[currying:] for each single-use function from  $ X \times Y $ to $Z$, there should be a single-use function from $X$ to $Y \Rightarrow Z$ which inputs a first argument and returns a representation of the  partially applied function;
\end{description}
Only in the presence of all of these operations can we speak of a function space, and the corresponding category can be called closed. (Composition can be obtained through evaluation and currying, so the essential operations are evaluation and currying.) The following example shows that the currying operation is not single-use, and therefore the space of single-use functions as defined in this section is not closed.

\begin{example}\label{ex:currying-not-single-use}
    Consider  the single-use function
    \begin{align*}
    f : \atoms \times \atoms \to 1 + 1 \qquad (a,b)  
    \mapsto \begin{cases}
        \text{result of test $a = $ Mark} & \text{if } b = \text{Eve} \\
        \text{result of test $a = $ John} & \text{otherwise}.
        \end{cases}
    \end{align*}
The currying of this function, is a new function which  maps a first argument $a \in \atoms$ to the partially applied function $f(a,\_)$. This currying is
\begin{align*}
    a \mapsto \begin{cases}
        b \mapsto b = \text{Eve} & \text{if } a = \text{Mark} \\
        b \mapsto b \neq \text{Eve} & \text{if } a = \text{John} \\
        b \mapsto \text{false} & \text{otherwise}
        \end{cases}
\end{align*}
Recall that in Example~\ref{ex:first-single-use-function-space} we showed that the space of single-use functions of type $\atoms \to 1 + 1$ can be represented as $1 + 1 + \atoms + \atoms$. If we use this representation,  then the currying of the function $f$  is not single-use, because we need to compare the input atom $a$ to two constants, Mark and John. If we use the representation from Example~\ref{ex:decision-tree-types}, then the corresponding type will be $\atoms \otimes (1 + 1)^2$, but the problems with currying will persist. \exampleend
\end{example}

For similar reasons, 
the function space, we proposed above, will also not support function composition, which means that it cannot be used to convert automata into single-use monoids, as we would like to do, since the resulting monoid would need to use function composition as its monoid operation.\footnote{This problem is solved in~\cite{bojanczykstefanski2020} and~\cite{stefanski-phd} in a different way, namely by showing that every orbit-finite monoid necessarily divides a single-use monoid, using a  Krohn-Rhodes construction. However, this construction is difficult and delicate, in particular it does not work for atoms that have more structure than equality alone. In contrast, the proposal that we give in this paper works for other kinds of atoms, as discussed in Section~\ref{sec:beyond-equality}. }

To solve the problems above, we will introduce a more refined type system, which is based on linear types.  The main idea is to pay more attention to type  in Example~\ref{ex:decision-tree-types}. In that definition, we describe a single-use function by specifying the first equality test that it makes, and then giving two descriptions of the functions that will be used in each of the two possible outcomes of the equality test. The main observation is that these two outcomes are mutually exclusive, and therefore we intend to use only one of the two descriptions. For this reason, we will use a  type constructor $\&$ that comes from linear logic. The intended meaning is  that an object of type $X \& Y$ consists of two objects, but with the ability to use only one of them. Since linear logic uses $\otimes$  for the  product that we have so far denoted by $\times$, we will also follow that convention. Using these two kinds of products, the appropriate type for Example~\ref{ex:decision-tree-types} will now become:
\begin{align*}
\coprod_{i \in \set{1,\ldots,k}} 
        \atoms \otimes ((\atoms^{k-1} \Rightarrow Y) \& (\atoms^{k-1} \Rightarrow Y))
          \quad + \quad 
    \coprod_{i, j \in \set{1,\ldots,k}} (\atoms^{k-2} \to Y) \& (\atoms^{k-2} \to Y).
    \end{align*}
Under this definition, the problems from Example~\ref{ex:currying-not-single-use} will be solved, at least for the particular type considered in that example. However, by introducing a new type constructor, we will have to redefine the single-use functions, and then we will have to give a representation of functions that allow this new type constructor, without incurring the need to add any other new type constructors. This is the subject of the next section.

\subsection{Linear types and single-use functions on them}
\label{sec:linear-types-and-single-use-functions}
As mentioned above, to solve the problems with single-use function spaces, we will consider a type system with two kinds of products, as in the following definition.
\begin{definition}[Linear types]\label{def:datatypes}
    A \emph{linear type} is any expression constructed from the atomic types $1$ and $\atoms$ using three\footnote{We set up our type system without using the multiplicative disjunction $\rotatebox[origin=c]{180}{\&}$ of linear logic -- morally, we take our inspiration from intuitionistic linear logic, rather than classical linear logic.} binary type constructors $+, \&$ and $\otimes$.
\end{definition}
In our linear types, it is only the products that are differentiated, while  $+$ comes in only one version. 
    Here is the intuitive explanation of the difference between the two kinds of products, following Girard~\cite[p.2]{girard1995advances}. Having a pair $x \otimes y$ is like having the ability of using both components $x$ and $y$. On the other hand, having a pair $x \& y$ is like having the ability to use one of the two components, at our choice, but not both at once. For example, the input type of the equality test will be $\atoms \otimes \atoms$ not $\atoms \& \atoms$, since the test will need to consume both arguments. This intuition can only go so far; for example, it is not entirely clear what ``our choice'' means. We revisit this intuition in  the appendix, where game semantics will be used to indicate who makes which choices.

We think of each linear type $X$ as representing a set $\sem X$, as defined below:
\begin{align*}
    \sem{1} = 1
\quad 
\sem{\atoms} = \atoms 
\quad 
\sem{X+Y} = \sem X + \sem Y 
\quad 
\sem{X \otimes Y} =
\sem{X \& Y} = \sem X \times \sem Y.
\end{align*}
All sets that arise in this way will be polynomial orbit-finite sets.
Note that the two kinds of product represent the same set, namely the set of pairs in the usual set-theoretic sense. 
However, the two type constructors will be  different, because different functions will be allowed to operate on them. As the expression goes, ``the proof of the pudding is in the eating''; in this case the pudding is the types and the eating is the functions.  

As we did in Section~\ref{sec:single-use-functions-over-polynomial-orbit-finite-sets}, the single-use functions will be defined in terms of prime functions and combinators. The combinators are the same, except that instead of $f_1 \times f_2$ we now have two ways of pairing functions, using $\otimes$ and $\&$. The prime functions are inherited from the previous system, with $\times$ understood as $\otimes$, together with a few new functions for $\&$, as described in Table~\ref{fig:prime-morphisms-with-with}. This is summarized in the following definition.

\begin{definition}[Single-use functions] The class of single-use functions is the least class of functions with the following properties:
    \begin{enumerate}
        \item It contains the functions from Tables~\ref{fig:prime-morphisms-without-with} and~\ref{fig:prime-morphisms-with-with}, with $\times$ in Table~\ref{fig:prime-morphisms-without-with} understood as $\otimes$;
        \item It is closed under composition, as well as under combining functions using  $+$, $\otimes$ and $\&$. 
    \end{enumerate}
\end{definition}

\begin{table}[h!]
    \centering
    \begin{tabular}{lll}
        \textbf{Function} & \textbf{Type} & \textbf{Definition} \\ \\
        diagonal  & $X \to X \& X$ & $x \mapsto x \& x$ \\
        first projection & $X \& Y \to X$ & $x \& y \mapsto x$ \\
        second projection & $X \& Y \to Y$ & $x \& y \mapsto y$ \\
        $\&$ distributes over $\otimes$ & $X \otimes (Y \& Z) \to (X \otimes Y) \& (X \otimes Z)$ & $x \otimes (y \& z) \mapsto (x \otimes y) \& (x \otimes z)$ \\
        $\&$ distributes over $+$ & $X + (Y \& Z) \to (X \& Y) + (X \& Z)$ & $\left\{\begin{tabular}{l}
        $x \& \text{left}(y) \mapsto \text{left}(x \& y)$\\
        $x \& \text{right}(z) \mapsto \text{right}(x \& z)$
        \end{tabular}\right.$ \\ \\ 
    \end{tabular}
    \caption{Prime single-use functions that involve $\&$.}
    \label{fig:prime-morphisms-with-with}
\end{table}

Formally speaking, a single-use function consists of an input linear type $X$, an output linear type $Y$, and a function between the sets $\sem X$ and $\sem Y$ that is generated using the prime functions and combinators from the above definition.  As was the case in Section~\ref{sec:single-use-functions-over-polynomial-orbit-finite-sets}, all single-use functions are  finitely supported.  Therefore, one can think of the single-use functions of type $X \to Y$ as being a subset of the set of all finitely supported functions from $\sem X$ to $\sem Y$. This subset is strict: as we will see, the space of single-use functions will be orbit-finite, unlike the space of all finitely supported functions. We will be thinking of the single-use functions as a category.

\begin{definition}[Category of single-use sets]\label{def:single-use-category}
    The category of single-use sets is:
    \begin{enumerate}
        \item The objects are linear types, as per Definition~\ref{def:datatypes}.
        \item Morphisms between types $X$ and $Y$ are single-use functions from $\sem X$ to $\sem Y$.
    \end{enumerate}
\end{definition}

In the very definition of the above category, there is a faithful functor to the category of  orbit-finite sets with finitely supported functions. This functor maps objects $X$ to their underlying sets $\sem X$, which are orbit-finite sets, and it maps morphisms to the corresponding functions. The functions seen to be finitely supported, and the functor is faithful  by definition, since the morphisms in Definition~\ref{def:single-use-category} are defined to be single-use functions.


The main technical result of this paper is that the category of single-use sets has function spaces, as stated in the following theorem.  The appropriate product will be $\otimes$, and not $\&$. Since the Cartesian product in our category is $\&$ and not $\otimes$, this means that the result we are targeting is symmetric monoidal closed with respect to $\otimes$, and not Cartesian closed. (The same situation will happen  for vector spaces later in this paper.)  Our theorem stops a bit short of saying that the category is monoidal closed, since several different elements of the function space might represent the same function. This non-uniqueness of representation can be overcome  by quotienting the function space  by an equivalence relation that identifies two elements if they represent the same single-use function; this is explained in Appendix~\ref{sec:quotient-category}. For now, we stick with non-uniqueness of representation, and we state the theorem as follows.

\begin{theorem}\label{thm:single-use-closed}
    Let $V$ and $W$ be objects (i.e.~linear types). There exists an object, denoted by  $\funspace V W$, and a morphism (i.e.~a partial single-use function)
    \begin{align*}
    eval : (\funspace V W) \otimes V \to  W
    \end{align*}
    with the following property. For every morphism
    \begin{align*}
    f : {X \otimes V} \to  W
    \end{align*}
    there is a (not necessarily unique) morphism
    \begin{align*}
    h :  X \to (\funspace V W)
    \end{align*}
    such that the following diagram commutes:
    \[
    \begin{tikzcd}
    X \otimes V 
    \arrow[r,"h \otimes id"]
    \arrow[dr,"f"']
    &
    (\funspace V W) \otimes V
    \arrow[d,"eval"] \\
    &
    W
    \end{tikzcd}
    \]
\end{theorem}

The above theorem is the main technical contribution of this paper, and its proof is presented in Appendix~\ref{sec:game-semantics}. The difficulty in the proof is finding a representation of the single-use  functions that is rich enough to capture all functions, but simple enough to be described by a linear type (in particular, the corresponding set will be orbit-finite). In  Section~\ref{sec:single-use-functions-over-polynomial-orbit-finite-sets}, when the types did not have $\&$, we could pull off a relatively simple construction, which was possible mainly due to the strong distributivity rules that allowed converting each type into a normal like $\atoms^{n_1} + \cdots + \atoms^{n_\ell}$. In the presence of $\&$, the distributivity rules are not as strong, and the way in which a single-use program can interact with its input is rather subtle. 
Our solution, and the technical core of this paper, is to use game semantics, appropriately extended to describe the type $\atoms$ and the operations on it that are allowed. This solution is presented entirely in the appendix.

\section{Beyond equality}
\label{sec:beyond-equality}

So far, we have only studied the case when the atoms are equipped with equality only. Consider now a more general case: let $\atoms$ be any relational structure, i.e.~any set equipped with relations (but not functions). For example, we could use the rational numbers with their linear order. Another example would be Presburger arithmetic, i.e.~the integers $\Nat$ together with addition $+$. Since we want to have relations only, we think of addition as a ternary relation $x + y =z$. 

The construction of the single-use category from Section~\ref{sec:single-use-sets} is generic enough so that it be generalized to any relational structure, and not just atoms with equality. (In fact, there is also a suitable generalization for function symbols, but we do not prove anything about it, so we do not talk about it.)

\begin{definition}[Single-use functions over a relational structure]\label{def:single-use-category-relational-structures}
    Let $\atoms$ be a relational structure. The \emph{single-use category over $\atoms$} is defined in the same way as in Definition~\ref{def:single-use-category}, except that the list of prime functions is extended with one prime function  $\atoms^n \to 1 + 1$ 
        for every $n$-ary relation in the vocabulary of $\atoms$. (Here, the power $\atoms^n$ uses $\otimes$.)
\end{definition}

Not only is the definition of the category generic, but the same is true for  the proof that function spaces exist. The generalized version of this theorem, stated below, makes only a relatively tame assumption, namely that the vocabulary has finitely many relations of each arity. For example, such vocabularies will arise if we start with some finite vocabulary, and then enrich it by creating a new relation symbol for every quantifier-free formula (modulo equivalence of quantifier-free formulas). In the following theorem, when speaking of a symmetric monoidal closed category, we mean that properties stated in the conclusion of Theorem~\ref{thm:single-use-closed}, in particular we do not require unique Currying.

\begin{theorem}\label{thm:single-use-closed-relational-structures}
    Consider a relational structure $\atoms$, in which for every $k \in \set{0,1,\ldots}$ there are finitely many relations of arity $k$. Then the single-use category over $\atoms$ is symmetric monoidal closed, with respect to the tensor product $\otimes$.
\end{theorem}

The theorem is proved in the same way as Theorem~\ref{thm:single-use-closed}.  The assumption on the vocabulary is used to ensure that in the corresponding game semantics, there are finitely available moves in any given moment, because the vocabulary can only be queried up to the maximal number of atoms in the input, due to the single-use restriction.

Note that this  theorem can be applied to any relational structure, including undecidable ones. Clearly, there must be some benefit from assuming that the structure has a decidable first-order theory, which means that there is an algorithm which checks if a first-order sentence is true in the structure. 
The benefit is that we can check if two morphisms are equal, as expressed in the following theorem. 

\begin{theorem}\label{thm:first-order-decidable}
    Consider a relational structure $\atoms$, in which for every $k \in \set{0,1,\ldots}$ there are finitely many relations of arity $k$. If this structure has a decidable first-order theory, then there is an algorithm for the following problem:
    \begin{itemize}
        \item {\bf Input.} Two morphisms, described by expressions using prime functions and combinators.
        \item {\bf Question.} Are they the same morphism?
    \end{itemize}
\end{theorem}

The above theorem gives us a reasonable category of single-use functions over a relational structure with a decidable first-order theory. This applies to structures such as Presburger arithmetic, or the real field $(\mathbb R, +, \cdot, \leq)$ with the field operations viewed as ternary relations. 
However, a decidable first-order theory is not the only property needed to ensure that the category is appropriate to automata. If we want to model automata and their decision procedures, we will also need to execute certain fixpoint algorithms, as explained in~\cite{bojanczyk_slightly2018}. This will be illustrated in the next section, where we prove that emptiness for automata, including two-way automata, is decidable under an additional assumption on the structure called \emph{oligomorphism}. This assumption is the standard assumption used to ensure that orbit-finiteness is meaningful. 


\section{Traced categories and two-way automata}
\label{sec:two-way-automata}

In this section, we show that the single-use category is traced with respect to co-product, and we use this fact to model deterministic two-way automata.
We begin by describing the usual trace construction in the category of sets and partial functions. 
Consider some partial function $
f : A + X \to B + X.$
For a number $k \in \set{0,1,\ldots}$ we can define a partial function of type $A \to B$, which is called the \emph{$k$-th iteration}, as  follows by induction on $k$. The $0$-th iteration is  completely undefined. For $k > 0$, the $k$-th iteration is defined as follows. First we apply $f$ to the input. If the output is undefined or from  $B$, then that is the final output of the $k$-th iteration. Otherwise, if the output is from $X$, then we apply the $(k-1)$-st iteration and return that output (which may be undefined). It is easy to see that the iterations are ordered by inclusion, i.e.~when viewed as binary relations they form an increasing chain. The limit of this chain, which is a partial function from $A$ to $B$ is defined to be the \emph{trace of $f$}. 

We will now show that the traces exist also in the single-use category, even if we take the general variant from Section~\ref{sec:beyond-equality} that arises from some relational structure. However, we will need to make a certain assumption on this structure. Call a structure $\atoms$  \emph{oligomorphic} if for every $k \in \set{0,1,\ldots}$, there  are finitely many orbits in $\atoms^k$ under the action of the automorphism group of $\atoms$.
Oligomorphism is the standard assumption in the theory of sets with atoms~\cite{bojanczyk_slightly2018}. In particular, ensures that the notion of orbit-finite set is meaningful. Examples of oligomorphic structures include: the atoms with equality only, the rational numbers with their linear order, and the Rado graph. Nonexamples include: the integers with their linear order, Presburger arithmetic, and the real field.

The main observation is that for single-use functions over an oligomorphic structure, the trace is achieved in finitely many steps.
\begin{lemma}\label{lem:traced-finite-iteration} Let $\atoms$ be an oligomorphic structure, let $A,B,X$ be linear types, and let 
    \begin{align*}
        f : \sem{A + X} \to \sem{B + X}
        \end{align*}
        be a single-use function over $\atoms$. There is some $k$ such that the trace of $f$ is equal to $f^{(k)}$.
\end{lemma}

Since the trace is achieved in finitely many steps, and it is easily seen to be constructed using operations on functions that preserve the single-use condition, the trace is also single-use. (In this section, we use single-use partial functions, which are defined to be single-use functions of type $X \to Y +1$, where $1$ represents the undefined value.) Therefore, the single-use partial functions over an oligomorphic structure form what is called a \emph{traced category}, with respect to $+$.

We now use the trace construction to model deterministic two-way automata. (The idea that traced categories are a natural setting for two-way automata was already noted in~\cite{hines2003categorical}.) We define a deterministic two-way automaton in the same way as a deterministic one-way automaton, except that the transition function now is a partial function from $\Sigma \otimes (Q + Q)$ to $Q + Q + 1$. 
The two copies of $Q$ represent entering the letter from the left or right, and the $1$ in the output represents acceptance.  The function is partial, and thus if we want to model it as a complete function, then there will be another $1$ in its output, representing rejection.

\begin{theorem}\label{thm:two-way-automata}
    Let $\atoms$ be an oligomorphic structure with a decidable first-order theory. Then the emptiness problem is decidable for single-use deterministic  two-way automata over $\atoms$.
\end{theorem}

The single-use restriction is crucial, since without it the emptiness problem becomes undecidable for all choices of $\atoms$ where the universe is infinite. The above theorem is the first known example of two-way model that has decidable emptiness and uses atoms that are not the equality atoms.

\section{Orbit-finite vector spaces}
As a complement to the single-use category described before, we present a second example of a symmetric monoidal closed category that  uses atoms. In this example, instead of restricting functions, we generalize them. This generalization is based on the  orbit-finitely spanned vector spaces that were introduced in~\cite{bojanczykKM21OrbitFiniteVector}. As mentioned in the introduction,  this section contains no new technical results beyond those from~\cite{bojanczykKM21OrbitFiniteVector}. 

To explain the intuitive reason behind the usefulness of vector spaces, consider finite (and not orbit-finite) sets. The number of functions from set with $n$ elements to a set with $m$ elements is $m^n$. However, if we consider linear maps instead of functions without any structure, then we get an exponential improvement, namely instead of $m^n$ we will have $m \cdot n$. 
This is because a linear map from a vector space of dimension $n$ to a vector space of dimension $m$ can be specified by giving an $m \times n$ matrix. As we show in this section,  in the orbit-finite world, the improvement is more pronounced: from infinite to finite.

For the sake of concreteness, all  vector spaces will be over the field of rational numbers. However, the choice of the field will not affect the results. The general idea is that the field has no atoms in it, and therefore it has a trivial action of atom automorphisms. 
\begin{definition}[Vector space with atoms]
    A \emph{vector space with atoms} is a set with atoms, equipped with a vector space structure, such that the vector space structure is equivariant, i.e.~vector addition is equivariant, and  scalar multiplication $v \mapsto cv$ is equivariant for every field element $c$.
\end{definition}

\begin{example} \label{ex:lina} For a set $X$, let us write $\Lin X$ for the vector space that consists of finite linear combinations of elements from $X$. When $X$ is the set of atoms, the resulting vector space $\Lin \atoms$ is a vector space with atoms, because it comes equipped with a natural action of atom automorphisms.  An example vector in this space is $
3 \cdot \text{John} + 2 \cdot \text{Eve} - 5 \cdot \text{Mark}.
$ \exampleend
\end{example}

When talking about vector spaces with atoms, one has to be careful about using bases. This is because finding a basis might require choice, and choice is not available in the presence of atoms. 
For example, consider space $\Lin \atoms$ from Example~\ref{ex:lina}, and its subspace $V$ which contains those vectors where all coefficients sum up to zero. This space is spanned (i.e.~generated using linear combinations) by the orbit-finite set which consists of all pairs $a - b$, where $a$ and $b$ are distinct atoms.
However, this spanning set is not a basis, since it contains linearly dependent vectors, for example John - Eve and Eve - John. Keeping only one of these vectors would require choice, which cannot be done using equivariant functions, and in fact this space has no basis that is equivariant~\cite[Example 6]{bojanczykKM21OrbitFiniteVector}. For this reason, the appropriate notion of finiteness is having an orbit-finite spanning set, which is not necessarily linearly independent. This leads to the following category.

\begin{definition}\label{def:orbit-finite-vector-space-category}
    The category of orbit-finitely spanned vector spaces is:
    \begin{enumerate}
        \item The objects are vector spaces with atoms that have an orbit-finite spanning set.
        \item The morphisms are equivariant linear maps.
    \end{enumerate}
\end{definition}

For the above category, we will use tensor product $\otimes$; this is the usual notion of tensor product for vector spaces. For example, if we take vector spaces $\Lin X$ and $\Lin Y$, i.e.~finite linear combinations of orbit-finite sets $X$ and $Y$ respectively, then the resulting tensor product is $\Lin (X \times Y)$. More generally, the tensor product of two orbit-finitely spanned vector spaces is also orbit-finitely spanned~\cite[Theorem VI.3]{bojanczykKM21OrbitFiniteVector}. 


\begin{theorem}\label{thm:orbit-finite-vector-space-closed} The category of orbit-finitely spanned vector spaces is symmetric monoidal closed, with respect to the tensor product.
\end{theorem}

A corollary of the above theorem is that, in the category of orbit-finitely spanned vector spaces, deterministic automata recognize the same languages as monoids. This is because the usual translation can be done, with a deterministic automaton of states $Q$ being mapped to a monoid with elements $\funspace Q Q$. This was already observed in~\cite[Theorem VIII.3]{bojanczykKM21OrbitFiniteVector}.

However, there are two limitations of the orbit-finitely spanned vector spaces. 

The first limitation is that the existence of function spaces is dependent on the choice of atoms. Theorem~\ref{thm:orbit-finite-vector-space-closed}  works when the atoms have equality only, and it also works when the atoms are equipped with a total order. This is because the dual spaces are orbit-finitely spanned in these cases, as proved in~\cite[Corollary VI.5]{bojanczykKM21OrbitFiniteVector}. However, the dual spaces are no longer orbit-finitely spanned for other choices of  atoms, such as the Rado graph, see~\cite[Example 9]{bojanczykKM21OrbitFiniteVector}. This is in contrast to the single-use category, where the existence of function spaces is independent of the choice of atoms. 

A second limitation is that it does not support two-way automata, unlike the single-use category (see Section~\ref{sec:two-way-automata}). This is because this category generalizes the orbit-finite category (i.e.~it admits a faithful functor from it), and in the orbit-finite category emptiness is undecidable for deterministic two-way automata~\cite[Theorem 5.3]{nevenFiniteStateMachines2004}. This precludes the kind of traced construction that we did in the single-use category. This issue appears already without atoms: the category of finite-dimensional vector spaces is not traced with respect to the sum $+$ of vector spaces.

\bibliography{bib}

\appendix

\newcommand{\invar}[1]{#1_{\mathrm{in}}}
\newcommand{\outvar}[1]{#1_{\mathrm{out}}}

\section{Game semantics}
\label{sec:game-semantics}

This section is devoted to the proof of Theorem~\ref{thm:single-use-closed}. To construct the function space $X \Rightarrow Y$, we use game semantics to identify a certain normal form of programs that compute single-use functions. The presentation in this section is self-contained, and does not assume any knowledge of game semantics. We base our notation on~\cite{abramsky2013semantics}.

Let us begin with a brief motivation for why game semantics will be useful.

While it is intuitively clear which functions should be allowed as single-use for simple types such as $\atoms \to 1+1$ or $\atoms \otimes \atoms \to \atoms + \atoms$, these intuitions start to falter when considering more complex types. How can one show that a function is \emph{not} single-use? If one were to use the definition of single-use functions alone, then one would need to rule out any possibility of constructing the function from the primes using the combinators, including constructions that use composition many times, and with unknown intermediate types. 

This is the reason why we consider game semantics. It will allow us to give  a more principled description of the intuition that pairs of type $X \otimes Y$ can be used on both coordinates, while pairs of type $X + Y$ can be used on a chosen coordinate only. The idea behind game semantics is to give the description in terms of an interaction between two players.  The two players are:
\begin{enumerate}
    \item System, who represents the function (we will identify with this player); and
    \item Environment, who supplies inputs and requests outputs of the function.
\end{enumerate}
One of the intuitions behind the setup is that if a type $X \& Y$ appears in the input of the function, then it is the System who can choose to use $X$ or $Y$, while if the type appears in the output, then it is the Environment who makes the choice. (In this paper, we consider functions of first-order types of the form $X \to Y$, where $X$ and $Y$ are linear types that do not use $\otimes$, and therefore there will be a clear distinction between input and output values.)
Before giving a formal definition of game semantics, we give simple example of the interaction.

\begin{example}\label{ex:amp-otimes-distr}
    Consider the two types 
    \begin{align*}
        X \otimes (Y \& Z) 
        \quad \text{and} \quad
        (X \otimes Y) \& (X \otimes Z).
    \end{align*}
    Among the prime functions in Table~\ref{fig:prime-morphisms-with-with}, we find distributivity in the direction $\rightarrow$, but not in the direction $\leftarrow$. We explain this asymmetry using the interaction between two players System and Enviroment.

    Let us first consider the interaction in the  direction $\rightarrow$. Player Environment begins be requesting an output. Since this output is of type $(X \otimes Y) \& (X \otimes Z)$, this means that Environment can choose to request either of the two types  $X \otimes Y$ and $X \otimes Z$. Suppose that Environment requests $X \otimes Y$. Now player System needs to react, and produce two elements: one of type $X$ and one of type $Y$. Both can be obtained from the input; for the second one player System can choose how to resolve the input $Y \& Z$ to get the appropriate value. 

    Consider now the interaction in the opposite direction $\leftarrow$. As we will see, player System will be unable to react to the behavior of player Environment, which will demonstrate that there is no distributivity in this direction. The problem is that player Environment can begin by requesting an element of type $X$, since the output type is $X \otimes (Y \& Z)$, while still reserving the possibility to request $Y \& Z$ in the future (because the tensor product $\otimes$ means that both output values need to be produced). To produce this element, player System will need  choose one of the two coordinates in the input type $(X \otimes Y) \& (X \otimes Z)$, and any of these two  choices will be premature, since player Environment can then request the opposite choice in the output type.  \exampleend
\end{example}

As illustrated in the above example, we will use a game to describe the possible behaviours of a function, as modelled by behaviours of player System. The game will be played in an arena, which will arise from the type of the function, and will tell us what are the possibilities for the moves of both players.  Here is the  outline the plan for the rest of this section. 
\begin{enumerate}
    \item For every two linear types $X$ and $Y$, we  define an arena, which is a data structure that describes all possible interactions between players Environment and System that can arise when running a  single-use  functions of type $X \to Y$;
    \item In each arena, we will be interested in the strategies of player System, i.e.~the ways in which System reacts to moves of Environment. We will show how such strategies can be interpreted as single-use functions: to a strategy we will assign a single-use function of type $X \to Y$ that is represented by this strategy. This mapping will be partial, i.e.~some ill-behaved strategies will not represent any functions. 
    \item We will show that the  set of strategies in an arena is  large enough to represent all single-use functions, but small enough to be orbit-finite. 
    \item We will then strengthen this: not only is the set of strategies orbit-finite, but it can be equipped with the structure of a linear type, such that both evaluation and Currying will be single-use functions.
\end{enumerate}

The arenas from the first step of the plan will  be defined in two sub-steps. We begin by defining arenas and strategies for functions that do not use the structure of the atoms, i.e. constants and equality tests. This will be a fairly generic definition, almost identical to the classical game semantics for linear logic. Then we will extend the definition to cover the extra structure. 

\subsection{Arenas and strategies without constants and equality tests}
\label{sec:arenas-without-constants-and-equality-tests}

We begin with a simpler version of the game semantics, in which the arenas and strategies will describe functions that are not allowed to perform equality tests, and are not allowed to use constants. These strategies will model functions such as the identity function $\atoms \to \atoms$, which directly passes its input to its output, but they  will not model the equality test $\atoms \otimes \atoms \to 1 + 1$, or the constant functions of type $1 \to \atoms$. The general idea is to use standard game semantics for linear logic, with an extra feature that we call \emph{register operations}. The register operations will be used to model the way in which atoms are passed from the input to output. For example, in the identity function,  Environment will  write the input atom into the register, and then player System will read the output atom from that register. The following definition of an arena is based on the definition from \cite[p.4]{abramsky2013semantics}, slightly
adapted for the context of this paper:
\begin{definition}[Arena] \label{def:arena}
    An \emph{arena} consists of:
    \begin{enumerate}
        \item A set of \emph{moves}, with each move having an assigned owner, who is either ``System'' or ``Environment'', and one of three\footnote{\label{footnote:read-write} In all arenas that we consider, the ``read'' moves will be owned by System and the ``write'' moves will be owned by Environment. Therefore, we could simplify the register operations and have just one, called ``read/write'', whose status is determined by its owner. 
        } register operations, which are ``none'', ``read'', or ``write''.
                \item A set of plays, which a set of finite sequences of moves that is closed under prefixes, and such that in every play, the owner of the first move is Environment, and then the owners alternate between the two players.
    \end{enumerate}
\end{definition}

An arena will correspond to a type. The inhabitants of that type, which will be  functions if the type is a functional type $X \to Y$, will be described by strategies in the arena. Such a strategy tells us how player System should react to the moves of player Environment. Intuitively speaking, in the case of a functional type, a strategy will say how the function reacts to requests in the output type and values in the input type. We will only be talking about strategies for player System, so from now on, all strategies will be for player System. The following definition corresponds to the definition from \cite[p.5]{abramsky2013semantics}:

\begin{definition}[Strategy]
    A strategy  in an arena  is a subset of plays in the arena, which: 
    \begin{enumerate}
        \item is closed under prefixes;
        \item\label{item:sys-ext} if the strategy contains a play $p$ that ends with a move owned by player System, then it also contains all possible plays  in the arena that extend $p$ with one move of player Environment;
        \item\label{item:env-ext} if the strategy contains a play $p$ that ends with a move owned by player Environment, then it contains exactly one play  in the arena that extends $p$ with one move of player System;
         \item there is some $k$ such that all plays in the strategy have length at most $k$;
        \item every ``read'' move is directly preceded by a ``write'' move (in particular a play cannot begin with ``read''), and every ``write'' move is either the last move, or directly succeeded by a ``read'' move.
    \end{enumerate}
\end{definition}

Conditions \ref{item:sys-ext} and \ref{item:env-ext}, which are standard in game semantics,  guarantee that the strategy only ``ends'' when 
Environment has no moves to play.  Let us now comment on the last two conditions, which are not standard.

The fourth condition is motivated by the idea that we study ``finite'' types,  and there will be no need for unbounded computations. 

The last condition will be called the  \emph{immediate read condition}. It ensures that there is matching between ``read'' and ``write'' moves in plays that do not end with write. Since ``write'' will always be owned by Environment, the immediate read condition will ensure a matching between ``write'' and ``read'' moves.

We now show how to associate to each linear type a corresponding arena, and also how to associate an arena to a function type $X \to Y$. This definition will be compositional, i.e.~it will arise through operations on arenas that correspond to the type constructors such as $1$ or $\otimes$. The arenas that we will construct so far will not be our final proposal, since the corresponding strategies will not be able to use constants or perform equality tests. This will be fixed in Section~\ref{sec:arenas-with-constants-and-equality-tests}, where a more complex arena will be defined for function types. Before giving formal definition, we begin with a simple example of the arena for the type $\atoms \to \atoms$.
\begin{example}\label{ex:identity-function-without-equality-tests-and-constants}
    We define an arena for the type $\atoms \to \atoms$. This arena will be rather impoverished, since the  only allowed strategy in it will correspond to the identity function. However, this is consistent with the stage that we are at, where we only consider functions that do not use constants or perform equality tests; for such functions the only possibility is the identity function.

     The arena will describe the following interaction between the two players: Environment  requests an output,  then System requests an input, then Environment grants the input, and finally player System grants the output by forwarding the input grant. 
    The arena is shown in the following picture: 
    \mypic{6}
    The methodology of drawing this picture will become clearer later on, as we define operations on arenas such as $\otimes$. For the moment we describe the arena without caring that it arises as a special case of some general construction. The arena has four moves: 
    \begin{center}
        \begin{tabular}{lll}
         move & owner & register operation \\
            \hline
            request output & Environment & none \\
            request input & System & none \\
            grant input & Environment & write \\
            grant output & System & read 
        \end{tabular}    
    \end{center}
    The set of plays is defined as follows. These are all sequences that begin with a move of player Environment, alternate between players, use each move only once, and have the following condition:  ``grant output'' can only be played after ``request output'', and likewise for ``grant input'' and ``request input''.  

    A quick inspection of the above definition reveals that the arena has a unique maximal play, where the moves are played in the order from the table, and all other plays are prefixes of this maximal play. Because of the uniqueness of responses, the set of plays is also a strategy. As mentioned at the beginning of this example, the strategy describes the identity function.  \exampleend
\end{example}

We hope that the above example explains some intuitions about how arenas describe types and strategies describe functions. We now give a formal definition. As mentioned before, this definition is compositional: we define arenas for the basic types $1$ and $\atoms$, and then we define operations on arenas that correspond to the type constructors  $+$, $\&$, and $\otimes$. We begin with the basic types.

\begin{definition}[Arenas for $1$ and $\atoms$]\label{def:arenas-without-atoms-or-functions} \ 
    \begin{enumerate}
        \item     The arena for the type $1$ is empty: there are no moves and the only play is the empty sequence. 
        \item The arena for type $\atoms$ has two moves, which must be played one after the other: first player Environment makes a move called ``request'' that has no register operation, and then player System responds with  a move called ``grant'' that has register operation ``read''.
    \end{enumerate}
\end{definition}

In the above definition, we only described the behaviour of $\atoms$ when viewed as an output type. To get the input type, where the players are swapped and read is swapped with write, we will use duality, which is another operation on arenas. This operation, together with other operations that correspond to the type constructors, are defined below.   

\begin{definition}[Operations on arenas]\label{def:composition-of-arenas}
    Let  $A$ and $B$ be arenas. We define the following arenas (see also Figure~\ref{fig:arena-constructors}):
        \begin{description}
            \item[$A+B$] The moves in this arena are the disjoint union of the moves of $A$ and $B$, with inherited owners and register operations, plus three extra moves: ``ask'' owned by Environment, and ``left'', ``right'' owned by System. 
            The plays are defined as follows. Player Environment begins with  an ask move, then System responds with a left or right move, and the remaining sequence is a play in the arena $A$ or $B$, depending on whether System chose left or right. 
            \item[$A \& B$] The moves in this arena are the disjoint union of the moves of $A$ and $B$, with inherited owners and register operations, plus three extra moves: ``acknowledge'' owned by System, and ``left'', ``right'' owned by Environment. 
            The plays are defined as follows. Player Environment begins by choosing left or right, then player System responds with an acknowledge move, and the remaining sequence is a play  in the arena $A$ or $B$, depending on whether Environment chose left or right.
            (This construction differs slightly from the one from \cite[Excercise~1.10]{abramsky2013semantics} -- this is because we want to keep it analogous to the construction for $A + B$.)
            \item[$A \otimes B$] The moves in this arena are the disjoint union of the moves of $A$ and $B$, with inherited owners and register operations. A play in this arena is any shuffle of plays in the two arenas $A$ and $B$. (A shuffle of two words is any word obtained by interleaving them, e.g.~shuffles of ``abc'' and ``123'' include ``a1b23c'' and ``12ab3c'').
            By Definition~\ref{def:arenas-without-atoms-or-functions}, we require that the owners of the move alternate in the 
            interleaved sequences. (This construction is based on \cite[p.7]{abramsky2013semantics}.)

            \item[$\bar A$] This is called the dual arena of $A$. The moves and plays are the same as in $A$, except the owners are swapped, and the ``read'' and ``write'' register operations are swapped.
        \end{description} 
\end{definition}

\begin{figure}
The arena $A+ B$:
            \mypic{7}
The arena $A \& B$: 
            \mypic{8}
The arena $A \otimes B$:
            \mypic{9}
    \caption{\label{fig:arena-constructors} Pictures of the operations on arenas. The picture for $\otimes$  is less useful than previous two, since the root node of the tree is not a player, but a node labelled by $\otimes$. The intuition is that the game is played in parallel on both arenas, and therefore a position in it can be visualized as a pair of positions in the two arenas.}
\end{figure}

Equipped with the above definitions, we  present our first attempt at assigning arenas to types. In the second item of the following definition, we use the name \emph{library-less}, because our final definition of the arena for a function type, as presented in the next section, will be equipped with an extra feature that will be called a library. 

\begin{definition} Let $X$ and $Y$ be linear types.
    \begin{itemize}
        \item The \emph{arena for  $X$} is defined by inductively applying the constructions from Definition~\ref{def:arenas-without-atoms-or-functions} and Definition~\ref{def:composition-of-arenas} according to the structure of the type.
        \item The \emph{library-less arena for $X \to Y$} is defined to be (dual of arena of $X$) $\otimes$ (arena of $Y$).
    \end{itemize}
\end{definition}

As discussed previously, our notion of arenas does not yet take into account the structure of the atoms, i.e.~the constants and equality tests. This will be fixed in the next section, by modifying the second item in the above definition. On the other hand, the arenas from that first item in the above definition, for linear types without function types, are already in their final form. 

In principle the construction from the second item in the above definition can be nested, and thus used to assign arenas to higher order types that can nest $\to$ with the other type constructors. This is how it is usually done in linear logic. However,  the construction that we will describe in the next section will  be less amenable to nesting, and will use it only  to describe functions between types that do not use $\to$.

\begin{example}
    The arena for the identity type $\atoms \to \atoms$ is the same as the arena from Example~\ref{ex:identity-function-without-equality-tests-and-constants}.
\end{example}

\subsection{Arenas and strategies with constants and equality tests}
\label{sec:arenas-with-constants-and-equality-tests}
In the Section~\ref{sec:arenas-without-constants-and-equality-tests}, we described arenas for functions that did not use the structure of the atoms, i.e.~constants and equality tests. We now show how these arenas can be extended to cover this structure. The general idea is to equip the arenas with an extra part, which we call the \emph{library},  that describes the allowed operations on the atoms. (The library as we present it here only contains functions for equality and constants, but in the proof of Theorem~\ref{thm:single-use-automata-relational-structures}, we can use a library that has other relations beyond equality. )

\begin{definition}[The library arena] \label{def:library-arena} The library arena and its parts are defined as follows, see Figure~\ref{fig:library-arenas} for pictures.
    \begin{enumerate}
        \item The \emph{constant choice arena} is the following arena $\atoms+1$ moves:
        first player System chooses an atom, then player Environment plays move with register operation ``write''. 
        \item The \emph{equality test arena} is an arena which the plays are as follows:
    \begin{enumerate}
        \item first player System plays a move with register operation ``read'';
        \item then player Environment plays an move with no register operation;
        \item then player System plays a move with register operation ``read'';
        \item then player Environment plays one of two moves, called $=$ and $\neq$, with no register operation.
    \end{enumerate}
    \item The \emph{library arena} is defined to be an arena that is obtained by applying $\otimes$ to infinitely many copies of the constant choice arena and infinitely many copies of the  equality test arena.
    \end{enumerate}
\end{definition}

\begin{figure}
The constant choice arena:
        \mypic{10}
The equality test arena:
    \mypic{11}
\caption{\label{fig:library-arenas} Pictures of the library arenas.  We use the convention that the register operations are in red, and the names of the moves, which have no other role than to distinguish them, are in black. 
Note that the first move in this arena is owned by System, and we assume in Definition~\ref{def:arena} that the first move is owned by Environment. This is because this arena, like all arenas in Definition~\ref{def:library-arena}, is not intended to be a stand-alone arena, but only as part of the bigger arena from Definition~\ref{def:arena-for-function-type} where the first move is indeed owned by Environment. } 
\end{figure}
       
The library arena is infinite. Taking the tensor product of infinitely many copies of the two arenas ensures that the library arena satisfies the following property, which corresponds to the $!$ operation from linear logic: 
\begin{align}\label{eq:bang-library-arena}
\text{library arena} 
\quad \equiv \quad 
\text{(constant choice arena)} \otimes 
\text{(equality test arena)} \otimes
 \text{(library arena)}.
\end{align}
In the above, $\equiv$ refers to isomorphism of arenas, which is defined in the natural way: this is a bijection between the moves, which is consistent with the owners, register operations and plays in the expected way.  Another property is that the library arena is isomorphic to a tensor product of itself: 
\begin{align}\label{eq:library-arena-isomorphism}
\text{library arena}
\quad \equiv \quad
\text{library arena} \otimes \text{library arena}.
\end{align}

We are now ready to give the final definition of arenas for functions between linear types, which takes into account the structure of the atoms.

\begin{definition}[Arena for a function type]\label{def:arena-for-function-type} For linear types $X$ and $Y$, the arena of $X \to Y$ is 
    \begin{align*}
    \text{(library arena)} \otimes \text{(dual arena of $X$)} \otimes \text{(arena of $Y$)}.
    \end{align*}
\end{definition}

This completes the game semantics of linear types and functions between them. We do not intend to give game semantics for higher order types, such as functions on functions etc. As a result, we will only be using the dual once, namely for the arena of the input type. Also, note that the read/write operations will be used in a restricted way, as announced in Footnote~\ref{footnote:read-write}, namely that the ``read'' moves will be owned by System and the ``write'' moves will be owned by Environment.  This is because the library arena has this property, the arena for $\atoms$ has this property, and all operations on arenas that we have defined preserve this property.

\subsubsection{Composition of strategies}
\label{sec:composition-of-strategies}
One of the main points about strategies in game semantics is that they can be composed.
The usual construction for the function type $X \to Y$ is to take the tensor product of the dual arena for $X$ and the arena for $Y$. 
Our construction is a bit more involved, since the arena that we use has a copy of the library arena. We now explain how to compose strategies in a way that accounts for the library arena. 

We begin by describing the usual construction for composing strategies, which we call \emph{shuffling and hiding}, see \cite[p.12]{abramsky2013semantics}, with a minor adaptation to our setting that has register operations. 

\begin{definition}[Shuffling and hiding]\label{def:shuffling-and-hiding}
    Let $A, B, C$ be arenas, and consider two strategies
    \begin{align*}
    \sigma_1 \text{ in the arena $A \otimes \bar B$}
    \qquad
    \sigma_2 \text{ in the arena $B \otimes C$}.
    \end{align*}
    The shuffling and hiding strategy for  $\sigma_1$ and $\sigma_2$, which is a strategy  in the  arena $A \otimes C$, is defined to be the set of plays    $p$  such that there exist plays $p_1 \in \sigma_1$ and  $p_2 \in \sigma_2$  with the following properties: 
\begin{itemize}
    \item the play $p$ satisfies the immediate read condition;
    \item the following two sequences of moves are equal: 
    \begin{enumerate}
        \item the subsequence of moves in $p_1$ that are from the arena $\bar B$;
        \item the subsequence of moves in $p_2$ that are from the arena $B$.
    \end{enumerate}
\end{itemize}
\end{definition}

Note that the set of moves in the arenas $\bar B$ is the same as the set of moves in the arena $B$ (the owners and register operations are changed),  and therefore the subsequences in items 1 and 2 above can be meaningfully compared.  The following lemma, whose straightforward proof is left to the reader,  shows that the above definition is well-formed.

\begin{lemma}
    The strategy described in Definition~\ref{def:shuffling-and-hiding} is a valid strategy.
\end{lemma}


We can now use shuffling and hiding to compose strategies in the arena for a function type that were defined in Definition~\ref{def:arena-for-function-type}.
Consider three linear types $X, Y, Z$, and two strategies, in the arenas for $X \to Y$ and $Y \to Z$, respectively.
 By unfolding the definition of arenas from Definition~\ref{def:arena-for-function-type}, these are strategies in the arenas 
\begin{align*}
\text{(library arena)} \otimes \overline{\text{arena for $X$}} \otimes \text{arena for $Y$} & \qquad \text{the arena for $X \to Y$, and}
\\
\text{(library arena)} \otimes \overline{\text{arena for $Y$}} \otimes \text{arena for $Z$} & \qquad \text{the arena for $Y \to Z$}.
\end{align*}
Using the construction from Definition~\ref{def:shuffling-and-hiding}, we can combine them into a single  strategy in 
\begin{align*}
\text{(library arena)} \otimes  \text{(library arena)}\otimes \overline{\text{arena for $X$}} \otimes \text{arena for $Z$}.
\end{align*}
Since the library arena is isomorphic to a tensor product of two copies of itself, the above strategy gives us a strategy in the arena 
\begin{align*}
    \text{(library arena)} \otimes   \overline{\text{arena for $X$}} \otimes \text{arena for $Z$} & \qquad \text{the arena for $X \to Z$.}
    \end{align*}

This strategy is defined to be the \emph{composition} of the original two  strategies. We write $\sigma; \tau$ for this composition. Thanks to \cite[Proposition~1.2]{abramsky2013semantics}, we know that the the usual
composition of library-free strategies is associative. It follows that our composition of library strategies is associative 
up to isomoprhism.

\subsection{Strategies as single-use functions}
In this section, we explain how a strategy in a function type $X \to Y$ can be interpreted as a single-use function from the set  $\sem X$ to the set $\sem Y$. This interpretation is partial, since some strategies will be considered inconsistent, and will not represent any function. We begin with an example which explains how  a strategy might be inconsistent. 

\begin{example}[Inconsistent strategy] \label{ex:inconsistent-strategy} Consider a strategy in the arena for  $1 \to \atoms \otimes \atoms$. We would like to think of this strategy of producing a pair of atoms. 
     However, an inconsistent strategy for player System could use different constants from the constant choice arena  depending on the order in which Environment requests the two output atoms in $\atoms \otimes \atoms$.  For example, the output could be (John, Mark) for one order, and (Eve, Mary) for some other order.  \exampleend
\end{example}

Below, we will formalize the intuition from the above example. We begin by developing some notation for describing plays.

\subsubsection{Some notation for describing plays}
\label{sec:some-notation-for-describing-plays}
Consider a move in the arena associated to a function type $X \to Y$. There are three kinds of moves here: moves from the library, moves from $X$, and moves from $Y$. Depending on where the moves come from, we will associate to each move an \emph{origin} as follows. 
\begin{itemize}
    \item In the case of moves which come from the library, define the  {origin} to be the copy of the constant choice arena or the equality test arena from which the move is taken.
    \item   In the case of moves which come from either $X$ or $Y$, define the origin to be the corresponding node in the syntax tree of $X$ or $Y$. Here, the syntax tree of a type is defined in the obvious way, see Figure~\ref{fig:syntax-tree}.
\end{itemize}

\begin{figure}
    \mypic{12}
    \caption{\label{fig:syntax-tree} The syntax tree of the type $((\atoms +1) + 1) \otimes ((\atoms + \atoms) \& \atoms)$.}
\end{figure}

Consider a play in the arena for $X \to Y$. Define a \emph{library call} to be a subsequence of the play that is obtained by choosing some copy of the constant choice arena or the equality test arena, and then returning all moves in the play that originate from this copy. A library call will have at most two moves if it originates from some copy of the constant choice arena (we use the name \emph{constant call} for such library calls), and at most four moves if it originates from some copy of the equality test arena (we use the name \emph{equality call} for such library calls). For each constant call, there is an associated at atom, which is the first move in the call. 

Consider a read move in the play. By the immediate read condition, this move is directly preceded by a write move. The write move will either originate from a leaf labelled by $\atoms$ in  the syntax tree of the input type $X$, or it will originate from some constant call. In the latter case, it will have a corresponding atom. Therefore, to each read move in a play we can associate either some leaf labelled by $\atoms$ in the syntax tree of the input type, or some individual atom constant. The associated object will be called the \emph{read value} of the read move.

Equipped with the above notation, we are ready to define which strategies  can  be interpreted as values or functions.

\subsubsection{Strategies in $1 \to X$ as values in $\sem X$}
\label{sec:strategies-in-1-to-x-as-elements-of-sem-x}
  We begin by describing strategies which represent values in the underlying set $\sem X$ of a linear type $X$, and then we lift this interpretation to functions.  To represent such values, we  use strategies not in the arena for $X$, but in the arena for $1 \to X$. This is because the second arena contains the library, which is needed to  produce  atom constants.

  Let $X$ be a linear type, and consider a strategy $\sigma$ in the arena for $1 \to X$. 
  In order to represent some value, this strategy must satisfy certain conditions that avoid the bad behaviour  from Example~\ref{ex:inconsistent-strategy}.

\begin{definition}[Strategy that represents a value]\label{def:consistency}
    Let $X$ be a linear type and consider a  strategy in  the arena for $1 \to X$. We say that the strategy uses a node from the syntax tree of $X$ if at least one play in the strategy uses a move that originates from this node. We say that the  strategy    \emph{represents a value} if it satisfies the following conditions: 
    \begin{enumerate}
        \item Consider a node in  the syntax tree of $X$ that is labelled by co-product $+$. Then at most one child of this node is  used by the strategy. 
        \item Consider a node  in the syntax tree of $X$ that  is labelled by $\atoms$. Then there is some atom $a$ such that for every play in the strategy, if the play contains a read move that originates from the node, then the read value of this move is $a$. 
    \end{enumerate}
\end{definition}

For a strategy that satisfies the two conditions in the above definition, we can also define the value $x \in \sem X$ that it represents.  This is done as follows.  For every node  in the syntax tree  of $X$ which is used by the strategy, we assign a value, which is called the value of the strategy at the node. This value at a node is in the set $\sem Y$, where $Y$ is the linear type that corresponds to the subtree of the node.  The definition of this value at a given node  is defined using a straightforward induction on the size of its subtree: 
\begin{enumerate}
    \item the value at a leaf labelled by $1$ is the unique element of $\sem 1$;
    \item the value at a leaf labelled by $\atoms$ is the atom from the second item in Definition~\ref{def:consistency};
    \item the value at a node labelled by $\otimes$ or $\&$ is the pair of values at the children;
    \item the value at a node labelled by $+$ is the value at  the unique child in the first item in Definition~\ref{def:consistency}.
\end{enumerate}
The above definition is well-formed, since one can easily see from the definition of strategies that if a node labelled by $+$ is used by the strategy, then at least one of its children is used, and if a node labelled by $\otimes$ or $\&$ is used, then both of its children are used. Finally, the value represented by a strategy is defined to be the value that is represented at the root node of the syntax tree. 

One can easily see that every value in the set $\sem X$ is represented by at least one strategy. This will also follow from a more general result, Lemma~\ref{lem:all-functions-represented} below.

\subsubsection{Strategies in $X \to Y$ as single-use functions}
We now move to the main point of this section, which is to show how to interpret a strategy in the arena for $X \to Y$ as a single-use function of type $X \to Y$. This is done by using the definition from for values in the previous section, and then lifting it to functions by using composition of strategies that was defined in Section~\ref{sec:composition-of-strategies}.

\begin{definition}\label{def:strategy-as-single-use-function}
    Let $X$ and $Y$ be linear types, and let $\sigma$ be a strategy in the arena $X \to Y$. We say that $\sigma$ \emph{represents} a function  
    \begin{align*}
    f : \sem{X} \to \sem{Y}
    \end{align*}
    if for every $x \in \sem X$, if  $x$ is the value  represented by some strategy $\tau$ in the arena for $1 \to X$, then the composition strategy $\tau;\sigma$ in the arena for $1 \to Y$ represents a value, and this value is equal to $f(x)$.
\end{definition}

Since each strategy in $1 \to X$ or $1 \to Y$ represents at most one value, it follows that a strategy can represent at most one function. This function is called the \emph{function represented by the strategy}.  We are now ready to state the main result of this section, which is that every single-use function is represented by at least one strategy.
\begin{lemma}\label{lem:all-functions-represented}
    Let $X$ and $Y$ be linear types, then every single-use function of type $X \to Y$ is represented by at least one strategy. 
\end{lemma}
\begin{proof}
    It suffices to show that all single-use prime functions can be represented as strategies,
    and  to lift the combinators $\circ$, $\otimes$, $\times$, $\&$ from functions to 
    strategies in a way that preserves the semantics. Let us start with the combinators:
    \begin{description}
        \item[$f; g$] We compose the corresponding strategies. The correctness of this construction is witnessed by the following claim. 
\begin{claim}\label{claim:composition-of-strategies-and-functions}    Let $X, Y, Z$ be linear types, and 
    let $\sigma_1$ and $\sigma_2$ be strategies in the arenas for $X \to Y$ and $Y \to Z$, respectively. If both of these strategies represent functions, then  $\sigma_1;\sigma_2$ also represents a function, and this function is equal to the composition of the functions represented by $\sigma_1$ and $\sigma_2$.
\end{claim}
\begin{proof}
    Let $\tau$ be some value consistent strategy in the arena for $1 \to X$. By associativity of strategy composition, we know that the strategies 
    \begin{align*}
    \tau;(\sigma_1;\sigma_2) \quad \text{and} \quad (\tau;\sigma_1);\sigma_2
    \end{align*}
    are equal, up to an isomorphism of the library arena. Since such an isomorphism does not affect value consistency or the represented value, either both or none are value consistent, and if they are value consistent, then they represent the same value. The right strategy is value consistent by the assumption of the lemma, and therefore the same is true for the left strategy, and the represented values are equal.
\end{proof}

        \item[$f_1 + f_2$] Suppose that functions $f_1$ and $f_2$ are represented by strategies $\sigma_1$ and $\sigma_2$. We join these two strategies into a new strategy $\sigma$ as follows.  Player Environment begins, and the only move available is to ``ask'' for one of the two disjuncts of the output type, which is the co-product of the output types of $f_1$ and $f_2$. Player System responds with a corresponding question on the input type. The only choice for player Environment is to select either ``left'' or ``right'' in the input type. Player System forwards this choice to the output type, and then uses the corresponding strategy.  Every play in this strategy has the following form: it begins with four moves that select a disjunct, and then it continues with a play in the strategy $\sigma_1$ or $\sigma_2$, depending on  the appropriate disjunct. From this property, one can show that the strategy represents the function $f_1 + f_2$.
        \item[$f_1 \& f_2$]  The construction is defined analogously to the one for $f_1 + f_2$.
        \item[$f_1 \otimes f_2$] Again, we use an analogous construction, in which player System plays in parallel the two strategies representing the functions $f_1$ and $f_2$. The corresponding property, which ensures that the strategy represents the function $f_1 \otimes f_2$, is that each play can be split into a shuffle of two plays, one for the strategy representing $f_1$ and one for the strategy representing $f_2$.
    \end{description}
    Let us now deal with prime functions. For the sake of brevity we only show the strategies for selected prime functions. The remaining prime functions can be implemented without using  the library, and are well known strategies corresponding to tautologies of affine logic:
    \begin{description}
        \item[Equality test of type $\atoms \otimes \atoms \to 1 + 1$]:
        For this function, System uses the library function for comparing atoms 
        and returns the result. The library is designed to implement this function, so the construction is straightforward, but we describe it here for the convenience of the reader:
\begin{itemize}
    \item Environment begins by  playing ``ask'' in $1+1$, this is the only choice
    \item System  plays  ``request'' in the left $\atoms$.
    \item Environment replies with ``grant'' which is a write move, this is the only choice.
    \item System plays the read move in atoms equality library function.
    \item Environment acknowledges, this is the only choice.
    \item System plays  ``request'' in the right $\atoms$.
    \item Environment plays ``grant'', this is the only choice.
    \item System plays the read move in the atoms equality function.
    \item Environment replies with either $=$ or $\neq$ in the library, these are the only choices.
    \item System forwards this choice to the output type.
\end{itemize}

        \item[Constant function of type $1 \to \atoms$] As in the previous item, this function is implemented by using the library, in this case the constant choice part of the library. 

        \item[Projection of type $ X \otimes Y \to X$] System uses the  shuffling and hiding strategy that was described in Section~\ref{sec:composition-of-strategies}. 

        \item[Distributivity  of type $X \otimes (Y \& Z) \to (X \otimes Y) \& (X \otimes Z)$] The  strategy from Example~\ref{ex:amp-otimes-distr}.  
    \end{description}
\end{proof}
One could also show that all represented functions are single-use. We omit this proof, since we will prove a stronger result in the next section. 

\subsection{The set of strategies as a linear type}
The purpose of this section is to show that the set of strategies in an arena for a function type $X \to Y$ can be described using some linear type, and furthermore the relevant operations on strategies, such as application and Currying, can be performed in a single-use way.

Before we do this we introudce the idea of a \emph{partially applied arena}.
Let $A$ be an arena, and $m$ a move. We define the partially applied arena ${m^{-1}}A$ 
to be set of sequences that follow $m$ in $A$, i.e:
\[m^{-1}A = \{s \ | \ ms \in A \}.\]
The partially applied arena also remembers two bits of information about its context: (a) wheather the current move belongs to the Environment or to the System
(b) wheter the previous move was a write move (so that we know if the next move has to be a write move). We finish our discussion with the observation 
that if moves $a, b$ belong to $\atoms$ (i.e. they are constant requests from the library), then $a^{-1}A = b^{-1}A$. 
For this reason we define the notation $\atoms^{-1}A$ for partially applying any of the atomic moves.

We are now ready the type for storing $k$-bounded strategies.  The following definition inputs an explicit bound $k$ for the length of the plays, but later on in this section will show that such a bound can be derived from the types by using an analysis of strategies.
\begin{definition}
    For a (partially applied) arena $A$, we define the type $\Strat(A, k)$ of $k$-bounded strategies on $A$ using the following induction:
\begin{itemize}
    \item $k=0$. This type is $1$ if the Environment owns $A$ and has no moves to play. Otherwise, this type is $\emptyset$.
          (Observe that $\emptyset$ is not a valid type in our language, so in the inductive step we are going to handle it explicitly --
          the main idea is that if its used with $+$ then it is going to be ignored,
          and if it is used with $\&$ then the entire type is going to be $\emptyset$.) 
    \item $k + 1$. This type is defined differently, depending on whether Environment or System owns $A$.
    \begin{itemize}
        \item Environment owns $A$. Let $m_1,\ldots,m_n$ be the set of possible first moves. (This set is finite.) The type is defined to be 
        \begin{align*}
            \Strat(m_1^{-1}A, k) \quad \&\  \cdots\  \& \quad  \Strat(m_n^{-1}A, k).
        \end{align*}
        Let us briefly mention how to handle edge cases: If at least one $\Strat(m_n^{-1}A, k)$ is equal to $\varnothing$, then 
        the entire type is equal to $\varnothing$. If there are no valid moves $m_i$, then the entire type is equal to $1$. 
        \item System owns $A$. Now we also consider two options, depending on whether the previous move by Environment was a write move. 
               If it was a write move, then let $m_1, \ldots, m_n$ be the current set of possible read moves that are not the first moves in any library 
               component. (This set is finite). Then our type is defined as follows:
               \begin{align*}
                \Strat(m_1^{-1}A, k) + \cdots + \Strat(m_n^{-1}A, k) + \underbrace{\Strat({\text{eq}}^{-1}A, k)}_{\substack{\textrm{Start a new call to}\\
                \textrm{atoms comparing}\\
                \textrm{library function}}}.
               \end{align*}
               The other case is when the previous move was not write. Then let $m_1, \ldots, m_n$ be the set of available moves that are not read moves. 
               (This set is finite again.) The strategy type is defined as: 
               \begin{align*}
                \Strat(m_1^{-1}A, k) + \ldots + \Strat(m_n^{-1}A, k) + \underbrace{\atoms \otimes \Strat(\atoms^{-1}A, k)}_{\substack{\textrm{Start a new call to}\\
                \textrm{atomic constants}\\
                \textrm{library function}}}.
               \end{align*}
               Again, let us mention the edge cases: If some $\Strat(m^{-1}A, k)$ is empty (including the cases where $m$ is the new library move, i.e. $\atoms$ or $\text{eq}$)
               then we simply skip it. If all $\Strat(m_i^{-1}A, k)$ are empty, then the entire type is empty.
    \end{itemize}
\end{itemize}
\end{definition}
The intuition behind this definition, which was already discussed in the main body of the paper, is as follows. If the current move belongs to Environment, then the strategy 
for System needs to be prepared for every possible move by the environment. However, eventually Environment 
will play only one move, so only one of those paths will materialize. This behaviour corresponds exactly to the connective $\&$, and is the reason why  it is used to make our single-use category symmetric monoidal closed.
When the system is about to move, then it needs to pick one of the possible moves. This corresponds to the behaviour of the connective $+$.

In order to find a linear type that could represent all functions in $X \to Y$, we are going to show 
that for every linear types $X$ and $Y$ there is a bound $k$, such that every single use function $f : X \to Y$
can be represented as a $k$-bounded strategy, i.e. as $\sigma \in \Strat(X \to Y, k)$. Before we do this, 
let us point out some potential problems. 

The main problem is that (as long as $Y \not \approx 1$), there is no universal bound $k$ for strategies in the arena for $X \to Y$
(even though each strategy is bounded by some $k$). The  reason is that System might ask an arbitrary number of irrelevant questions, 
i.e. questions that use the equality test on two constants. Interestingly, the irrelevant questions might appear as a result of compositions of strategies:
\begin{example}[Irrelevant questions from composition]\label{ex:irrelevant-questions-from-composition}
    Consider the  following two functions: 
    \begin{enumerate}
        \item the function of type  $1 \to \atoms \otimes \atoms$ that outputs the pair (Mark, John) for its unique input;
        \item the equality test of type $\atoms \otimes \atoms \to 1 + 1$.
    \end{enumerate}
    Like all single-use functions, the above two functions have natural representations that do not ask irrelevant questions.
    However, if we apply the composition construction from Section~\ref{sec:composition-of-strategies} to these two strategies, we will get a strategy that asks the irrelevant question of whether Mark and John are equal. \exampleend
\end{example}

The idea of irrelevant questions captures the intuition of why strategies might be arbitrarily long, but it is not the only reason for this. Another reason is that System might use constants to start an arbitrary number of parallel calls to the equality test library function, 
never completing any of them. This is why we introduce the following definition:
\begin{definition}[Well-formed strategy]
    We say that a strategy is \emph{well-formed} if it never starts a new call to the equality test library function, by passing a constant as its first argument.
\end{definition}

\noindent
Well-formed strategies are bounded, as shown by the following lemma:
\begin{lemma}\label{lem:well-formed-bounded}
    For every linear types $X$ and $Y$ there is a universal bound $k$, such that all well-formed strategies in $X \to Y$ are bounded by $k$. 
\end{lemma}
\begin{proof}
    We start by defining the dimension of a type to be the maximal number of $\atoms$ that the type can store at the same time. This is defined inductively, 
    where $\dim(1) = 0$, $\dim(\atoms) = 1$, and the dimension of complex types is computed as follows:
    \[ \begin{tabular}{ccc}
        $\dim(X + Y) = \dim(X \& Y) = \max(\dim(X), \dim(Y))$ & and & $\dim(X \otimes Y) = \dim(X) + \dim(Y)$
    \end{tabular}
    \]
    This is because in case of $X \otimes Y$ both $X$ and $Y$ exist at the same times, and in case of $X + Y$ and $X \& Y$ the values $X$ or $Y$
    only one of them can materialize. 

    The key observation is that in a well-formed strategy, every constant granted by the environment is either compared with an atom from input, 
    or moved to the output. It follows that in a well-formed strategy, the system can ask for at most $\dim(X) + \dim(Y)$ constants
    (because, due to the immediate read rule, each constant has to be used). It is now not hard to see that the lengths of well-formed strategies 
    in $X \to Y$ are bounded by some $k$ that depends only on the types $X$ and $Y$.
\end{proof}

As illustrated by Example~\ref{ex:irrelevant-questions-from-composition}, well-formedness is not preserved under composition of strategies. 
For this reason, we need that every  strategy that is not well-formed  can be turned into a well-formed one in a way that preserves its semantics (later we will also 
show that this can be done in a single-use way). Before we do this, we would like to assume a weaker property of the possibly non-well-formed strategy 
of the input:
\begin{definition}[Read-consistent strategy]
    A strategy is called \emph{read-consistent} if, after requesting an atomic value from the environment, it uses the same move 
    to consume this value in all the branches. This request can be made either through a move in the constant
    library component or a "request" move in an $\atoms$ on the input. Here is a visual representation:
    \picc{read-consistant-ex} 
    (The squiggly lines represent more than one move). Here System requests the atom John from the environment. 
    Environment might respond immediately with write (as in the bottom branch) or it might wait some time 
    (as in the middle or top branch). The read-consistency property requires that the system uses the same move to read 
    the value John in all of those branches (i.e. $M_1 = M_2 = M_3$). 
\end{definition}

The advantage of read-consistency over well-formedness is that it is preserved under compositions of strategies. 
(One can show this directly from \ref{def:shuffling-and-hiding} -- the key idea is to compare the system's response to the immediate write,
with the system's response to every other possible write). Using this observation, we can show that read-consistent strategies are complete:
\begin{lemma}\label{lem:all-read-consistent}
    For every single use function $f : X \to Y$, there is a read-consistent strategy in the arena $X \to Y$ that represents it. 
\end{lemma}
\begin{proof}
    The proof is analogous to the proof of Lemma~\ref{lem:all-functions-represented}. As well-formedness is preserved
    under the combinators $\circ$, $\otimes$, $\&$, $+$. (The proof for $\otimes$, $\&$, and $+$ is immediate.)
    Moreover, it is not hard to see that all prime functions can be realized by read-consistent strategies. 
\end{proof}

Next, let us show that every read-consistent strategy can be made well-formed in a way that preserves its semantics
(thus proving completeness of well-formed strategies). 
\begin{lemma}\label{lem:all-well-formed}
    For every read-consistent strategy in the arena for $X \to Y$,
    there is a well-formed strategy in the same arena that represents the same function.
\end{lemma}
\begin{proof}
    Consider a system's request for a constant, in some branch of $\sigma$. 
    We say that this request is \emph{offending} if System's intention (according to read-consistenty) is to use
    the constant to start a new call to the equality test library function.
    Let us now show how to remove an offending request from a strategy in a way that does not change the represented function
    (without introducing any new offending requests). This is enough to prove the lemma,
    because we can repeat this construction for all offending moves. (And thanks to read-consistency, we know that 
    if a strategy contains no offending moves, then it is well-formed.)

    Let ``request Alice'' be the offending move. Here is how we remove it:
    \custompicc{offending-elim-1}{0.4}
    In the above picture,  $T'$ represents  the strategy obtained by inlining the atom Alice in $T$. 
    Specifically we need to remove all moves that are used to complete the 
    equality test library function that was initiated with the atom Alice.
    In order to do this, we identify all atom requests by the system, 
    whose value the system intends to compare with Alice.
    (In each of $T$ there can be at most one such move.)
    These could be of two kinds.
    \begin{enumerate}
        \item 
    The first kind is a request for another 
    constant. We deal with it in the following way:
    We ask what would happen if the Environment immediately returned a write move. 
    In this case, we know that System would respond with the completion of the equality test. 
    Then, we consider what would happen  when Environment immediately returns the result of the comparison -- 
    either $=$ or $\neq$ and replace the entire subtree with the appropriate move. 
    For example, we replace the following subtree with $T_{\neq}$:
    \custompicc{atom-inline}{0.55}
    \item We now consider the second kind of offending move, when System requests an atom from the input.
    In this case, we look at the entire subtree and replace each completion to equality call 
    with a single-move ``read and compare with Alice'', which represents a new fast-track 
    call to the equality test library function, that reads an atom and compares it with a constant
    (we will show how to get rid of this in the next paragraph).
    \custompicc{input-constant-inline}{0.55}

    This leaves us with showing how to get rid of the ``read and compare with Alice'' move. 
    We replace each such move with a usual equality test initialization (a read move). 
    Then, as long as Environment does not play the `acknowledge' move from the equality test, 
    we play as in the original strategy. As soon as the Environment plays `acknowledge', 
    we play ``request Alice'' to the constant library component. Then, until the environment 
    plays ``write Alice'', we continue to play as in the original strategy. Once the environment
    plays ``write Alice'', we play the ``complete equality test'' move, and then we continue 
    to play according to the original strategy. (This is possible, because the responses 
    to ``complete equality test'', are the same as the responses to ``read and compare with Alice''). 
    Here is a picture:
    \picc{elimnate-fast-track}
    \end{enumerate}
    This completes the proof of Lemma~\ref{lem:all-well-formed}. 
\end{proof}
By Lemmas~\ref{lem:all-well-formed}~and~\ref{lem:well-formed-bounded} we know that every function $X \to Y$ 
can be represented by a strategy bounded by $k$ that depends only on $X$ and $Y$. Thanks to this observation 
we are finally ready to define the linear type for function spaces, which is the main subject of Theorem~\ref{thm:single-use-closed}.
\begin{definition}[The function space type] Let $X$ and $Y$ be linear types. 
    We define $X \Rightarrow Y$ to be equal to be the type  $\Strat(X \to Y, k)$ of $k$-bounded strategies, where $k$ is the bound from Lemma~\ref{lem:well-formed-bounded}.
\end{definition}

It remains to define evaluation and Currying as single-use functions. To do this, 
let us introduce first a couple of auxiliary functions.
The following lemma, whose straightforward proof we omit, shows that we can transform a strategy into a value. 
\begin{lemma}\label{lem:evaluate-to-values}
    For every linear $X$, there is a single-use function of type $(1 \Rightarrow X) \to X$ that returns the value represented by a strategy (see  Section~\ref{sec:strategies-in-1-to-x-as-elements-of-sem-x}) if it exists, and otherwise returns some unspecified  value.
\end{lemma}

The unspecified value in the above lemma is used to ensure that the function is total, since our single-use functions are total. We can simply fix some value $x \in \sem X$, which is constructed using constants, and return this value for all strategies that do not represent any value. We will use this approach to representing partial functions later in this section as well.

Next, we show that we can compose strategies in a single-use way. 
\begin{lemma}\label{lem:compose-single-use-linear-type}
    Let $X \Rightarrow Y$ and $Y \Rightarrow Z$ be linear types, then there is a single use function:
    \[\text{compose} : (X \Rightarrow Y) \otimes (Y \Rightarrow Z) \to (X \Rightarrow Z)\]
    such that if $\sigma_1$ and $\sigma_2$ are strategies that represent functions $f_1$ and $f_2$ respectively,
    then $\text{compose}(\sigma_1, \sigma_2)$ represents the composition of $f_1$ and $f_2$.
\end{lemma}
\begin{proof}
    The proof has two steps:
    \begin{itemize}
        \item First, we show that for every $k, l$ there is a function
              $\Strat(X \to Y, k) \otimes \Strat(Y \to Z, k) \to \Strat(X \to Z, k + l)$, 
              that computes the composition of the strategies (according to the \ref{def:shuffling-and-hiding}). 
              This function can be implemented (by induction on $k + l$) by simulating the two strategies 
              and passing around the moves in $Y$. The construction will turn out to be single-use 
              thanks to that the hiding operation is an injection on 
              the shuffled strings see \cite[Covering Lemma]{abramsky2013semantics}.
        \item Next we show that for every $l$ there is a single-use function 
              $\Strat(X \to Y, l) \to \Strat(X \to Y, k)$, where $k$ is the 
              constant from Lemma~\ref{lem:well-formed-bounded}.
              The function works by removing the offending requests from the 
              strategy, as described in the proof of Lemma~\ref{lem:all-well-formed}. The key observation that makes this construction single-use is that we can use the 
              distributivities of $+$ and $\&$ to push atoms from the node 
              of a tree to all of its children:
              \[
                \begin{tabular}{ccc}
                    $ \atoms \otimes (T_1 \& \ldots \& T_n) \to \atoms \otimes T_1 \& \ldots \& \atoms \otimes T_n$ & and &
                    $ \atoms \otimes (T_1 +  \ldots + T_n) \to \atoms \otimes T_1 + \ldots + \atoms \otimes T_n$
                \end{tabular}
              \]
    \end{itemize}
\end{proof}

The final helper function transforms values to strategies. This is done in the following lemma, whose straightforward proof we omit.
\begin{lemma}\label{lem:to-strat}
    There is a single-use function of type $X \to (1 \Rightarrow X)$, 
    that returns a strategy whose semantics is equal to its input.
 \end{lemma}

We are now ready to define Currying and evaluation, and thus  complete the proof of Theorem~\ref{thm:single-use-closed}.  We begin with the evaluation function, which is defined as follows: 
    \[
    \begin{tikzcd}
     X \otimes (X \Rightarrow Y) 
        \ar[d, "{\text{function from Lemma~\ref{lem:to-strat} on first coordinate}}"]
    \\
        (1 \Rightarrow X) \otimes (X \Rightarrow Y)
        \ar[d, "{\text{composing function from Lemma~\ref{lem:compose-single-use-linear-type}}}"]
        \\ 
        1 \Rightarrow Y   
        \ar[d, "{\text{function from Lemma~\ref{lem:evaluate-to-values}}}"]
        \\ 
        Y    
    \end{tikzcd}
    \]
    Now we consider the Currying of  a single-use function
\begin{align*}
f : Z \otimes X \to Y.
\end{align*}
By Lemma~\ref{lem:all-well-formed} we know that this single-use function is represented by some element $\hat f$ of the  set represented by the type $Z \otimes (X \Rightarrow Y)$. We use this element to define Currying of $f$ as follows:
\[
    \begin{tikzcd}
     Z  
        \ar[d, "{\text{extend the input with $\hat f$}}"]
    \\
    Z \otimes (Z \otimes X \Rightarrow Y)
        \ar[d, "{\text{apply function from Lemma~\ref{lem:to-strat} to first argument}}"]
        \\ 
        (1 \Rightarrow Z) \otimes (Z \otimes X \Rightarrow Y)   
        \ar[d, "{\text{lift first argument to function that copies argument from $X$}}"]
        \\ 
        (X \Rightarrow X \otimes Z)  \otimes (Z \otimes X \Rightarrow Y)
        \ar[d, "{\text{composing function from Lemma~\ref{lem:compose-single-use-linear-type}}}"]
        \\ 
        X \Rightarrow Y 
    \end{tikzcd}
    \]
These functions are single-use, as compositions of single-use functions. They also satisfy the required properties in a symmetric monoidal closed category, which can easily be verified using the properties proved in the lemmas that were used to define the functions. This completes the proof of  Theorem~\ref{thm:single-use-closed}.

\section{A quotient construction}
\label{sec:quotient-category}
A drawback of Theorem~\ref{thm:single-use-closed} is that the function space $\funspace V W$ can contain different representations of the same function; this will mean that Currying is not unique. To overcome this issue, we use a simple quotient construction. 
Define a \emph{partial equivalence relation} to be a relation that is symmetric and transitive, but not necessarily reflexive.  This is the same as (complete) equivalence relation on some subset. We will use a partial equivalence on the function space $X \Rightarrow Y$ to: (1)
 remove objects that do not represent any morphism; (2) identify two objects if they represent the same morphism. After such a quotient, the function space will have unique representations for functions.

\begin{definition}
    The \emph{quotiented single-use category} is: 
    \begin{itemize}
    \item Objects are pairs (linear type $X$, equivariant partial equivalence relation on $\sem X$);
    \item Morphisms between objects $(X,\sim_X)$ and $(Y,\sim_Y)$ are single-use functions from $\sem X$ to $\sem Y$ such that the domain of the function is contained in the domain of $\sim_X$, and equivalent inputs are mapped to equivalent outputs.
    \end{itemize}
\end{definition}

The quotiented single-use category is also equipped with a tensor product $\otimes$ on its objects.
\begin{theorem}
    The quotiented single-use category, equipped with the tensor product $\otimes$, is a monoidal closed category, i.e.~it satisfies the conclusions of Theorem~\ref{thm:single-use-closed}, but, furthermore, the morphism $h$ is unique.
\end{theorem}

\section{Beyond equality}
\label{sec:beyond-equality-appendix}

As mentioned in the main body of the paper, the proof of Theorem~\ref{thm:single-use-closed} extends without any difficulty to the case of relational structures. We only prove the decidability result from Theorem~\ref{thm:first-order-decidable}.

\begin{proof}[Proof of Theorem~\ref{thm:first-order-decidable}]
    When defining the category of single-use functions, we viewed it as a restriction of a larger category, in the morphisms were equivariant functions.  We begin by describing the generalization of this larger category to the case where the parameter $\atoms$ is some relational structure. The idea is to use first-order definable functions as the general analogue of equivariant functions. 

\begin{definition}[Category of first-order definable functions]
    Let $\atoms$ be a relational structure. The \emph{category of first-order definable functions over $\atoms$} is defined as follows.
    \begin{itemize}
        \item The objects are polynomial sets over $\atoms$, i.e.~finite disjoint unions of powers of $\atoms$.
        \item The morphisms are the  first-order definable functions. Here, a function 
        \begin{align*}
            f : \atoms^{n_1} + \cdots + \atoms^{n_k} \to \atoms^{m_1} + \cdots + \atoms^{m_\ell}
            \end{align*}
            is called  \emph{first-order definable} if for every $i \in \set{1,\ldots,k}$ and  $j \in \set{1,\ldots,\ell}$ the set 
            \begin{align*}
            \setbuild{ (x,y) \in \atoms^{n_i} \times \atoms^{m_j}}{f(\text{$x$ in the $i$-th component}) = \text{$y$ in the $j$-th component}}
            \end{align*}
            can be defined by a first-order formula over the vocabulary of $\atoms$.
    \end{itemize}
\end{definition}
The above is easily seen to be a category, since first-order definable functions  can be composed. 
If the underlying structure $\atoms$ has a decidable first-order theory, then the category described above is reasonably tame, in particular one can decide if two morphisms are equal.

The single-use category admits a faithful functor to the above category. In particular, this implies that the problem of deciding if two morphisms are equal is decidable in the single-use category. 
\end{proof}

\section{Two-way automata}

\begin{proof}[Proof of Lemma~\ref{lem:traced-finite-iteration}]
    We only need the weaker assumption, which is that the function $f$ is finitely supported. Consider some set of atoms that supports $f$. This same set of atoms will also support all functions $f^{(k)}$, and it will also support their domains. In particular, the domains of these functions will form a decreasing sequence of subsets of $\sem{A + X}$, with all subsets having the same support. Since an orbit-finite set can have finitely many subsets with a given support, this sequence must terminate in finitely many steps.
\end{proof} 

\begin{proof}[Proof of Theorem~\ref{thm:two-way-automata}]
    The transition function in the automaton tells us how it behaves for a single input letter. We now extend this function to input strings.    For every input string $w$, we will define  a single-use function 
    \begin{align*}
       \delta_w :  Q+Q \to Q + Q + 1
    \end{align*}
    that describes the behaviour of the automaton on $w$. 
    
    This is done by induction on the length of the input string. Suppose that we have already defined the functions for  input strings $w_1$ and $w_2$.
    We want to compose them into a function for the string $w_1 w_2$. This function will need to account for multiple crossings of the head on the boundary between $w_1$ and $w_2$, as described in the following picture: 
    \mypic{13}
    Taking care of these crossings is exactly what is accomplished by the trace. Indeed, consider the function 
        \begin{align*}
            f : \myunderbrace{Q+Q}{$A$} + \myunderbrace{Q+Q}{$X$} \to \myunderbrace{Q+Q + 1}{$B$} + \myunderbrace{Q+Q}{$X$},
        \end{align*}
        where $A$ represents entering $w_1 w_2$ from either the left or right, $B$ represents exiting $w_1 w_2$ from either the left or right or accepting,  and $X$ represents the boundary between $w_1$ and $w_2$. The trace of $f$ is exactly the function that we need, in particular it will be undefined for inputs that lead to an infinite run.

        The definition of $\delta_w$ described above describes a function 
        \begin{align*}
        \text{input strings} \to \text{single-use functions of type }Q+Q \to Q + Q + 1.
        \end{align*}
        This function is a monoid homomorphism, for a suitably defined  monoid structure on the image. We use the name \emph{transition monoid} for the output monoid. We can view the transition monoid  as a linear type, namely the function space 
        \begin{align*}
        \funspace {Q+Q} {Q + Q + 1}.
        \end{align*}

        Like any linear type, the transition monoid represents an orbit-finite set. This set is equipped with a monoid operation. Although this monoid operation is not single-use (because tracing is not single-use, as an operation on function spaces), it is still an equivariant operation. Therefore, we have an orbit-finite monoid. We can now use a standard fix-point algorithm on orbit-finite sets~\cite{bojanczyk_slightly2018} to compute a representation of the sub-monoid that is generated by the images of the single letters. Finally, we can check if this sub-monoid contains an accepting element, i.e.~a transition function that maps the initial configuration (initial state at the left of the word) to an accepting state. 
\end{proof}

\section{Proof of Theorem~\ref{thm:orbit-finite-vector-space-closed}}
 We apply the standard proof that the category of vector spaces (without atoms and orbit-finiteness) is symmetric monoidal closed. We then observe that: (1) both evaluation and Currying are equivariant; (2) the resulting space is orbit-finitely spanned. 
    
    For the first observation (1), we simply note that the construction of the function space, as well as  the evaluation and Currying morphisms, are constructed using the language of set theory, and therefore they will necessarily be equivariant~\cite[Equivariance Principle]{bojanczyk_slightly2018}.

    Let us now explain the second observation (2), about the function space being orbit-finitely spanned. This uses a non-trivial result from~\cite{bojanczykKM21OrbitFiniteVector} which says that orbit-finitely spanned spaces are closed under duals. 
    The space $\funspace V W$ is defined to be the space of finitely supported linear maps from $V$ to $W$, with the natural vector space structure. This can be seen as a subspace of a larger space, call it $\funspacenonfs V W$, which contains all linear maps, not necessarily finitely supported. A standard result in linear algebra is that $\funspacenonfs V W$ is linearly isomorphic to $ \funspacenonfs {V \otimes W} 1$, where $1$ is the field. This isomorphism is equivariant, and it preserves the property of being finitely supported. Therefore, this gives us an equivariant linear isomorphism between $\funspace V W$ is isomorphic and $ \funspace {V \otimes W} 1$. The latter space is the dual of the orbit-finitely spanned space $V \otimes W$, and therefore it is orbit-finitely spanned by~\cite[Corollary VI.5]{bojanczykKM21OrbitFiniteVector}.

    We would like to remark that a similar result as Theorem~\ref{thm:orbit-finite-vector-space-closed} was observed in~\cite[Theorem 3.8]{przybyłek2024note}, but using the smaller category of vector spaces that admit an orbit-finite basis.

\end{document}